\documentclass[10pt]{amsart}
\UseRawInputEncoding
\usepackage{amsfonts}
\usepackage{amssymb}
\usepackage{amsmath,amsthm}
\usepackage{amscd}
\usepackage{graphics}
\usepackage{amsmath,amssymb}
\usepackage{pdfsync}
\usepackage[british]{babel} 
\usepackage{mathrsfs}
\usepackage{enumitem}
\usepackage{stmaryrd}
\usepackage[colorlinks,linkcolor=blue,citecolor=blue,urlcolor=blue]{hyperref}

\newcommand{\beq}{\begin{equation}}
\newcommand{\eeq}{\end{equation}}
\newcommand{\bea}{\begin{eqnarray}}
\newcommand{\eea}{\end{eqnarray}}
\newcommand{\nn}{\nonumber}
\newcommand\noi{\noindent}
\newcommand{\tbf}{\textbf}

\newcommand{\ti}{\textit}
\newcommand{\rd}{\mathrm{d}}
\newcommand{\al}{\alpha}
\newcommand{\be}{\beta}

\newcommand{\bk}{\begin{cases}}
\newcommand{\ek}{\end{cases}}
\newcommand{\bs}{\boldsymbol}

\newcommand{\f}{\frac}

\newcommand{\vs}{\vspace{3mm}}
\newtheorem{definition}{Definition}
\newtheorem{proposition}{Proposition}
\newtheorem{theorem}{Theorem}
\newtheorem{corollary}{Corollary}
\newtheorem{lemma}{Lemma}
\theoremstyle{definition}

\newtheorem{remark}{\textbf{Remark}}

\newtheorem{example}{\textbf{Example}}

\usepackage{enumerate}
\usepackage{float}
\usepackage{cancel}
\usepackage{hyperref}
\usepackage{mathtools}
\usepackage{afterpage}

\allowdisplaybreaks

\begin{document}

\title[Separability, Superintegrability and Haantjes Geometry]{Classical Multiseparable Hamiltonian Systems, Superintegrability  and Haantjes Geometry}
\author{Daniel Reyes Nozaleda}
\address{Departamento de F\'{\i}sica Te\'{o}rica, Facultad de Ciencias F\'{\i}sicas, Universidad
Complutense de Madrid, 28040 -- Madrid, Spain \\ and Instituto de Ciencias Matem\'aticas, C/ Nicol\'as Cabrera, No 13--15, 28049 Madrid, Spain}
\email{danreyes@ucm.es, daniel.reyes@icmat.es}
\author{Piergiulio Tempesta}
\address{Departamento de F\'{\i}sica Te\'{o}rica, Facultad de Ciencias F\'{\i}sicas, Universidad
Complutense de Madrid, 28040 -- Madrid, Spain \\  and Instituto de Ciencias Matem\'aticas, C/ Nicol\'as Cabrera, No 13--15, 28049 Madrid, Spain}
\email{piergiulio.tempesta@icmat.es, ptempest@ucm.es}
\author{Giorgio Tondo}
\address{Dipartimento di Matematica e Geoscienze, Universit\`a  degli Studi di Trieste,
piaz.le Europa 1, I--34127 Trieste, Italy.}
\email{tondo@units.it}

\date{September 02, 2021}


\begin{abstract}
We show that the theory of classical Hamiltonian systems admitting separating variables can be formulated in the context of ($\omega, \mathscr{H}$) structures. They are symplectic manifolds endowed with a compatible Haantjes algebra $\mathscr{H}$, namely an algebra of (1,1)-tensor fields with vanishing Haantjes torsion.  A special class of coordinates, called Darboux-Haantjes coordinates, will be constructed from the Haantjes algebras associated with a separable system. These coordinates enable the additive separation of variables of the corresponding Hamilton-Jacobi equation.

We shall prove that a multiseparable system admits as many $\omega\mathscr{H}$ structures as separation coordinate systems. In particular, we will show that a large class of multiseparable, superintegrable systems, including the Smorodinsky-Winternitz systems and some physically relevant  systems with three degrees of freedom, possesses multiple Haantjes structures.

\end{abstract}

\maketitle

\tableofcontents

\section{Introduction}
The prominence of integrable models in many areas of pure and applied mathematics and theoretical physics has motivated, in the last decades, a resurgence of interest in the algebraic and geometric structures underlying the notion of integrability. The study of the geometry of Hamiltonian integrable systems has a long history, dating back to the classical works by Liouville, Jacobi, St\"{a}ckel, Eisenhart, Arnold, etc. The approaches proposed in the literature are intimately related with the problem of the determination of suitable coordinate systems guaranteeing the additive separation of the Hamilton-Jacobi (HJ) equation. 

Recently, many new ideas coming from differential and algebraic geometry, topology and tensor analysis, have contributed to the formulation of important approaches such as the theory of bi-Hamiltonian systems \cite{Magri78}, the Lenard-Nijenhuis geometry \cite{MM1984} and the theory of Dubrovin-Frobenius manifolds \cite{D94}. These theoretical developments shed new light on the multiple connections among integrability, topological field theories, singularity theory, co-isotropic deformations of associative algebras, etc.
Besides, several integrable models, both classical and quantum ones, have  recently been discovered, in particular in the domain of superintegrability. 

Superintegrable systems are a special class of integrable systems which possess surprisingly rich algebraic and geometric properties  \cite{MF}, \cite{Nekh}, \cite{MPW}, \cite{PW2011}, \cite{TWR}, \cite{TTW}. Essentially, they  possess more independent integrals of motion than degrees of freedom. If a system with  $n$ degrees of freedom admits $2n-1$ functionally independent integrals of motion (the maximal number allowed), it is said to be \textit{maximally superintegrable}; if it admits $n+1$ integrals, then is minimally superintegrable.

Among the most famous examples of superintegrable models we mention the classical harmonic oscillator,  the Kepler potential, the Calogero-Moser potential, the Smorodinsky-Winternitz systems and the Euler top.  The presence of ``hidden'' symmetries, expressed by integrals which are second (or higher) degree polynomials in the momenta, usually allows us to determine the dynamical behaviour of superintegrable models. In the maximal case, the bounded orbits are closed and periodic \cite{Nekh}. As is well known, the phase space topology is also very rich: it can be described in terms of a symplectic bifoliation, determined by the standard Liouville-Arnold invariant fibration \cite{AKN} of Lagrangian tori, and complemented by a co-isotropic polar foliation \cite{Nekh}, \cite{Fasso}. 

Although  the present work focuses on classical systems, we point out that quantum superintegrable systems also possess many relevant properties. Exact quantum solvability of a Hamiltonian system, related to the existence of suitable Lie algebras of raising and lowering operators, could be regarded as the quantum analogue of the  classical notion of maximal superintegrability \cite{TTW}.

A fundamental class of integrable models is the \textit{separable} one: these models are characterized by the fact that one can find at least a system of canonical coordinates in which the corresponding Hamilton-Jacobi equation takes the additively separated form \begin{equation} \label{eq:W}
W= \sum_{k=1}^n W_k =\sum_{k=1}^n \int {p_k(q' _k; H_1, \ldots, H_n)_{\mid _{H_i=a_i} }dq'_k}.
\end{equation}

The problem of finding separating variables for integrable Hamiltonian systems has  been extensively investigated. 
In 1904, Levi-Civita proposed a test which permits to establish whether a given Hamiltonian separates in an assigned coordinate system \cite{L}. Another important result, due to Benenti \cite{Ben80}, states that
a family of Hamiltonian functions $\{H_i\}_{1\le i\le n}$ are separable in a set of
canonical coordinates $(\boldsymbol{q},\boldsymbol{p})$ if and only if they
are in {\rm separable  involution}, i.e. they satisfy the relations
\begin{equation} \label{eq:SI}
\{H_i,H_j\}_{\vert k}=\frac{\partial H_i}{\partial q^k}
\frac{\partial H_j}{\partial p_k}-\frac{\partial H_i}{\partial p_k}
\frac{\partial H_j}{\partial q^k}=0 \ , \quad 1\le k \le n
\end{equation}
where no summation over $k$ is understood. However, such theorem as well as the Levi--Civita test are not constructive,  and do not allow us to determine a complete integral of the Hamilton--Jacobi equation. 

A constructive approach for the determination of separating variables was given by Sklyanin  \cite{Skl} within the framework of Lax systems.
The Hamiltonian functions $\{H_i\}_{1\le i\le n}$ are separable in a set of
canonical coordinates $(\boldsymbol{q},\boldsymbol{p})$ if there exist $n$ suitable equations,  called the  Jacobi-Sklyanin separation equations for $\{H_i\}_{1\le i\le
n}$, having the form
\begin{equation} \label{eq:Sk}
\Phi_i(q^i,p_i; H_1,\ldots,H_n)=0 \qquad    {\rm  det}\left[\frac{\partial \Phi_i}{\partial H_j }\right] \neq 0, \quad i=1,\ldots,n \ .
 \end{equation} 
 These equations allow us to construct a solution of the HJ equation. 
In fact, by solving eq. (\ref {eq:Sk})
with respect to  $p_k=\frac{\partial W_k}{\partial q_k }$, we get
the additively separated form \eqref{eq:W}.

Nevertheless, the three above-mentioned criteria of separability are not intrinsic: in order to be applied, they require the explicit knowledge of the local chart $(\boldsymbol{q},\boldsymbol{p})$. To overcome such a drawback, in the last decades the modern theory of separation of variables (SoV) has been conceived in the context of symplectic and Poisson geometry; in particular, the bi-Hamiltonian theory has offered a fundamental geometric insight into the theory of integrable systems \cite{FP,MM1984}.

The main purpose of the present work is to establish a novel relationship between the theory of separable Hamiltonian systems and the geometry of an important class of tensor fields, the Haantjes tensors, introduced in \cite{Haa} as a relevant, natural generalization of the notion of Nijenhuis tensors \cite{Nij,Nij2}. The class of Nijenhuis tensors plays a significant role in differential geometry and the theory of almost-complex structures, due to the celebrated Newlander-Nirenberg theorem \cite{NN}.

Our approach is based on the notion of $\omega \mathscr{H}$ manifolds, introduced in \cite{TT2016prepr} by analogy with the theory of $\omega N$ manifolds \cite{MM1984,FP} for finite-dimensional Hamiltonian systems (see also \cite{FeMa}, \cite{MFrob} and \cite{MGall13} for a treatment of integrable hierarchies of PDEs). Essentially, an $\omega \mathscr{H}$ manifold is a symplectic manifold endowed with an algebra $\mathscr{H}$ of (1,1) tensor fields with vanishing Haantjes torsion, which are \textit{compatible} with the symplectic structure. Under the hypotheses of the \textit{Liouville-Haantjes} (LH) theorem proved in \cite{TT2016prepr},  a non-degenerate Hamiltonian system is completely integrable in the Liouville-Arnold sense if and only if it admits  a $\omega \mathscr{H}$ structure.

In our context, \emph{Haantjes chains} represent in the Haantjes framework the generalization of the notion of Lenard-Magri chain \cite{MagriLE} and of generalized Lenard chain \cite{FMT, MLenard} defined previously for quasi-bi-Hamiltonian systems \cite{MTlt}, \cite{MTltC}. By means of these structures, one obtains a complete description of the integrals of the motion of a system in terms of the associated Haantjes operators.

The problem of SoV can also be recast and studied in our approach.  Precisely, as stated in Theorem  \ref{th:SoVgLc} below, if an integrable system admits a \textit{semisimple} $\omega \mathscr{H}$ structure, one can derive  a set of coordinates, that we shall call the \textit{Darboux-Haantjes (DH) coordinates}, representing separation coordinates for the Hamilton-Jacobi equation associated with the system. In these coordinates, the symplectic form takes a Darboux form, and the operators of the Haantjes algebra take all simultaneously a diagonal form.
As we will show in Theorem \ref{th:multi}, and in the examples of Sections \ref{sec:construction} and \ref{sec:maximal},   multiseparable systems possess different  Haantjes structures associated in a nontrivial way with their separation coordinates. 

In the  study of separable Hamiltonian systems, the theory of $\omega N$ manifolds has proved to be a powerful tool. Since  any semisimple Haantjes algebra admits a Nijenhuis generator, we deduce that the $DH$ coordinates in the Haantjes scenario are also Darboux-Nijenhuis (DN) coordinates of the $\omega N$ theory. 

In the \textit{semisimple case}, we shall prove that there is a one-to-one correspondence between a given $\omega\mathscr{H}$ manifold and an equivalence class of $\omega N$ structures (see Section 3.7). Interestingly enough, given a semisimple $\omega\mathscr{H}$ manifold, one can find Nijenhuis  generators which fulfill both the algebraic and the \emph{differential} compatibility conditions, required by the $\omega N$ theory.

However, from a general, theoretical point of view, the two theories are not equivalent. Indeed, in the \textit{non-semisimple} case, there are integrable systems, as the Post-Winternitz system (discussed in \cite{TT2016prepr}), which admit a non-Abelian $\omega \mathscr{H}$ structure endowed with three Haantjes generators. In other words, there exist $\omega \mathscr{H}$ manifolds with several generators.  In these cases, an alternative description in terms of a standard $\omega N$ structure is not available, since in the $\omega N$ approach the existence of a unique Nijenhuis tensor $N$ is assumed.

This article is organized into two parts. In the first one, including Sections 2 and 3, for the sake of self-consistency we briefly summarize the notions necessary for the study of multiseparable systems. Precisely, in Section \ref{sec:NH}, the basic definitions concerning Nijenhuis and Haantjes tensors are proposed; in Section \ref{sec:integrable}, the Haantjes geometry is reviewed. In particular, the notions of Haantjes algebras, $\omega \mathscr{H}$ manifolds and Darboux-Haantjes coordinates are revised. In the second part, starting from Section 4, we shall propose the original results of our work. In Section \ref{sec:separation}, we propose the main theorem concerning the existence of Haantjes structures for separable systems. In Section \ref{sec:Drach}, as a direct application of the theory previously developed, we solve the problem of SoV for a family of Drach-Holt type systems, that were previously considered to be non-separable. Interestingly enough, the new separating variables we found are defined in the full phase space.  In Section \ref{sec:multiseparable}, this theorem is extended to the case of  multiseparable (and superintegrable) models. We propose in Section \ref{sec:lift} a novel geometric construction: a lift of operators from the configuration space $Q$ of dimension two to $T^{*}Q$, which generalizes the standard Yano lift \cite{YI1973}. By means of our procedure, a Haantjes operator can be lifted into another  Haantjes operator (unlike the Yano lift, which only preserves Nijenhuis operators, but not the Haantjes ones).
Section \ref{sec:construction} is devoted to the study of the Haantjes structures for the  Smorodinsky-Winternitz systems in the plane, whereas Section \ref{sec:anisotropic} deals with the study of the anisotropic oscillator. In Section \ref{sec:maximal}, the $\omega\mathscr{H}$ manifolds associated with certain important  multiseparable systems in three dimensions are determined. Future research perspectives are discussed in the final Section \ref{sec:future}.

\section{Nijenhuis  and Haantjes operators} \label{sec:NH}

\label{sec:1}

 In this Section, we review some basic algebraic results concerning the theory of Nijenhuis and Haantjes tensors. For a more complete treatment, see
the original papers~\cite{Haa,Nij}  and the related ones~\cite{Nij2,FN}.
\subsection{Geometric preliminaries}

Let $M$ be a real, differentiable $n$-dimensional manifold and $\boldsymbol{L}:TM\rightarrow TM$ a smooth $(1,1)$ tensor field, i.e.,~a field of linear operators on the tangent space at each point of $M$. In the following, all tensors will be assumed to be smooth.
\begin{definition}
The
 \textit{Nijenhuis torsion} of $\boldsymbol{L}$ is  the vector-valued $2$-form defined by 
\begin{equation}
 \label{eq:Ntorsion}
\mathcal{T}_ {\boldsymbol{L}} (X,Y):= \boldsymbol{L}^2[X,Y] +[\boldsymbol{L}X,\boldsymbol{L}Y]-\boldsymbol{L}\Bigl([X,\boldsymbol{L}Y]+[\boldsymbol{L}X,Y]\Bigr),
\end{equation}
where $X,Y \in TM$ and $[ \ , \ ]$ denotes the Lie bracket of two vector fields.
\end{definition}
In local coordinates $\boldsymbol{x}=(x^1,\ldots, x^n)$, the Nijenhuis torsion can be written as the skew-symmetric  $(1,2)$ tensor field
\begin{equation}
\label{eq:NtorsionLocal}
(\mathcal{T}_{\boldsymbol{L}})^i_{jk}=\sum_{\alpha=1}^n\biggl(\frac{\partial {\boldsymbol{L}}^i_k} {\partial x^\alpha} {\boldsymbol{L}}^\alpha_j -\frac{\partial {\boldsymbol{L}}^i_j} {\partial x^\alpha} {\boldsymbol{L}}^\alpha_k+\Bigl(\frac{\partial {\boldsymbol{L}}^\alpha_j} {\partial x^k} -\frac{\partial {\boldsymbol{L}}^\alpha_k} {\partial x^j}\Bigr) {\boldsymbol{L}}^i_\alpha \biggr)\ ,
\end{equation}
which possesses  $n^2(n-1)/2$ independent components.
\begin{definition}
 The \textit{Haantjes torsion} of $\boldsymbol{L}$ is the vector-valued $2$-form defined by
\begin{equation}
 \label{eq:Haan}
\mathcal{H}_{\boldsymbol{L}}(X,Y) := \boldsymbol{L}^2\mathcal{T}_{\boldsymbol{L}}(X,Y)+\mathcal{T}_{\boldsymbol{L}}(\boldsymbol{L}X,\boldsymbol{L}Y)-\boldsymbol{L}\Bigl(\mathcal{T}_{\boldsymbol{L}}(X,\boldsymbol{L}Y)+\mathcal{T}_{\boldsymbol{L}}(\boldsymbol{L}X,Y)\Bigr).
\end{equation}
\end{definition}

The skew-symmetry of the Nijenhuis torsion implies that the Haantjes torsion is also skew-symmetric.
Its local expression in explicit form is
\begin{eqnarray}
 \label{eq:HaanExCoord}
(\mathcal{H}_{\boldsymbol{L}})^i_{jk}&=&  \sum_{\alpha=1}^n
\biggl(-2  (\boldsymbol{L}^3)^i_\alpha \partial_{\mbox{[}j} \boldsymbol{L}^\alpha_{k\mbox{]}}+ (\boldsymbol{L}^2)^i_\alpha\Bigl( \partial_{\mbox{[}j} (\boldsymbol{L}^2) ^\alpha_{k\mbox{]}}+4  \sum_{\beta=1}^n\boldsymbol{L}^\beta_{\mbox{[}j} \partial_{\lvert \beta\rvert } \boldsymbol{L}^\alpha_{k\mbox{]}}\Bigr) \\
\nn &-&2 \boldsymbol{L}^i_{\alpha} \Bigr(\boldsymbol{L}^\beta_{\mbox{[}j} \partial_{\lvert \beta\rvert } (\boldsymbol{L}^2)^\alpha_{k\mbox{]}} +  (\boldsymbol{L}^2)^\beta_{\mbox{[}j}\partial_{\lvert \beta \rvert } (\boldsymbol{L}) ^{\alpha}_{k\mbox{]}}\Bigr)
+(\boldsymbol{L}^2)^\alpha_{\mbox{[}j}\partial_{\lvert \alpha \rvert } (\boldsymbol{L}^2) ^i_{k\mbox{]}}
\biggr) \ .
\end{eqnarray}
Here, for the sake of brevity, we have used the notation $\partial_j := \frac{\partial}{\partial x^j}$; it is understood that the indices between square brackets are to be skew-symmetrized, except those in $\lvert \cdot\rvert $.
 In \cite{TT2021}, the following notion was proposed.
\begin{definition}
A Haantjes (Nijenhuis)   operator is an operator field whose  Haantjes (Nijenhuis) torsion identically vanishes.
\end{definition}

\subsection{General properties of Haantjes operators}

First, we shall consider some specific cases for which the construction of the Nijenhuis and Haantjes torsions is very simple.

\begin{example}\label{ex:H2}
Let $dim~M=2$. Any operator field $\bs{L}:TM \to TM$ is a Haantjes operator. This can be proved by a straightforward calculation.
\end{example}

\begin{example}
Let dim $M=n$, $n\geq 2$ and $\boldsymbol{L}: TM\to TM$ be an operator field. Assume that in a suitable local coordinate chart $(x^1,\ldots, x^n)$ the operator $\boldsymbol{L}$ takes the diagonal form 
\begin{equation}
\boldsymbol{L}(\boldsymbol{x})=\sum _{i=1}^n l_{i }(\boldsymbol{x}) \frac{\partial}{\partial x^i}\otimes \rd x^i \label{eq:Ldiagonal} \ .
\end{equation} 
Then the Haantjes torsion of $\boldsymbol{L}$ identically vanishes.
\end{example}

Another interesting source of Nijenhuis and Haantjes operators is Classical Mechanics. Precisely, given a system of point masses in the $n$-dimensional affine Euclidean space, the inertia tensor of this system is a Haantjes tensor, whereas the planar inertia tensor is a Nijenhuis one \cite{TT2021}.

As is well known (see for instance~\cite{GVY}),  given  an invertible Nijenhuis operator, its inverse is also a Nijenhuis operator. The same property holds true for a Haantjes operator.

A crucial restriction in the Nijenhuis geometry is that, in general, the product of a Nijenhuis operator with an arbitrary $ \mathcal{C}^{\infty}(M)$-function  is no longer a Nijenhuis operator.
Instead, this is the case for Haantjes operators: therefore, they allow us to define new, interesting algebraic structures, as we shall see in Section \ref{sec:integrable}. The theory of these structures is based on the following 
\begin{proposition} [\cite{BogCMP}, \cite{BogI}] \hfill

i) Let  $\boldsymbol{L}$ be an operator  field. The following identity holds
\begin{equation} \label{eq:LtorsionLocal}
\mathcal{H}_{f \boldsymbol{I}+g \boldsymbol{L}}(X,Y)=g^4\, \mathcal{H}_{ \boldsymbol{L}}(X,Y),
\end{equation}
where $f,g \in C^{\infty}(M)$ are real functions and $\boldsymbol{I}$ denotes the identity operator on $TM$.

ii) Let $\boldsymbol{L}$ be a Haantjes operator. For any  polynomial in $\boldsymbol{L}$, with coefficients $a_{j}\in C^\infty(M)$, the associated Haantjes tensor vanishes, i.e.

\begin{equation}
\mathcal{H}_{\boldsymbol{L}}(X,Y)= \bs{0} \ \Longrightarrow \
\mathcal{H}_{(\sum_j a_{j} (\boldsymbol{x}) \boldsymbol{L}^j)}(X,Y)= \bs{0}.
\end{equation}
\end{proposition}

As proved in Ref. \cite{TT2021}, a Haantjes operator generates a cyclic Haantjes algebra (i. e., a cyclic algebra of Haantjes operators) over the ring of smooth functions on $M$. Cyclic Haantjes algebras will play a special role in our theory, as we shall clarify in the coming sections.

\section{Integrable frames, Haantjes algebras and $\omega \mathscr{H}$ manifolds}
\label{sec:integrable}
\subsection{Integrability}

In order to formulate our approach to separability, we shall review the relationship between Haantjes geometry and integrability. A more detailed treatment as well as the proofs of the statements reviewed here are available in Refs. \cite{TT2021}, \cite{TT2016prepr}. 

We start recalling that a \textit{reference frame} is a set of $n$ vector fields $\{Y_1,\ldots,Y_n\}$ satisfying the following property: given an open set  $U\subseteq M$,   $\forall ~\boldsymbol{x}\in U$ the frame represents a basis of the tangent space $T_{\boldsymbol{x}}U$.  Given two frames  $\{X_1,\ldots,X_n\}$ and $\{Y_1,\ldots,Y_n\}$, assume that $n$ nowhere vanishing smooth
functions $f_i$ exist, such that
\[
 X_i= f_i(\boldsymbol{x}) Y_i \ , \qquad\qquad i=1,\ldots,n \ .
\]
Then we shall say that the two frames are \textit{equivalent}. 
Let  $\{U, (x^1,\ldots,x^n)\}$ be a local chart of $U$. The frame formed by the vector fields $\left\{\frac{\partial}{\partial x^1}, \ldots, \frac{\partial}{\partial x^n}\right\}$ will be said to be a \textit{natural} frame.
\begin{definition}\label{def:Iframe}
A reference frame equivalent to a natural frame will be said to be \emph{integrable}.
\end{definition}
In the forthcoming considerations, given an operator $\boldsymbol{L}$, we shall denote by  $Spec(\boldsymbol{L}):= \{ l_1(\boldsymbol{x}), \ldots, l_s(\boldsymbol{x})\}$, $s\in \mathbb{N}\backslash\{0\}$,  the set of the pointwise distinct eigenvalues of  $\boldsymbol{L}$, assumed by default to be \emph{real}. The distribution of all the generalized eigenvector fields associated with the eigenvalue $l_i=l_i(\boldsymbol{x})$  will be denoted by
\begin{equation}
 \label{eq:DisL}
 \mathcal{D}_i: = \ker \Bigl(\boldsymbol{L}-l_i\boldsymbol{I}\Bigr)^{\rho_i}, \qquad i=1,\ldots,s
 \end{equation}
where $\rho_i \in \mathbb{N}\backslash \{0\}$ is the Riesz index of $l_i$ (which by hypothesis will always be taken to be independent of $\boldsymbol{x}$). The value $\rho=1$ characterizes the \textit{proper eigen-distributions}, namely the eigen-distributions of proper eigenvector fields of $\bs{L}$.

\begin{definition}\label{def:Kdiag}
An operator field $\boldsymbol{L}$ is \textit{semisimple}
if in each open neighborhood $U\subseteq M$ there exists a reference frame formed by proper eigenvector fields of $\boldsymbol{L}$. Moreover, $\boldsymbol{L}$ is \textit{simple} if all of its eigenvalues are pointwise distinct, namely $l_i (\boldsymbol{x}) \neq l_j(\boldsymbol{x})$, $i, j=1,\ldots,n$, $ \forall \boldsymbol{x}\in M$.
\end{definition}
A frame of proper eigenvectors will be said to be a proper eigen-frame of $\boldsymbol{L}$. If the frame contains generalized eigenvectors, it will be said to be a generalized eigen-frame.

\begin{definition}\label{def:mI}
Given a set of distributions $\{\mathcal{D}_i, \mathcal{D}_j, \ldots, \mathcal{D}_k \}$ of an operator $\bs{L}$, we shall say that such distributions  are \textit{mutually integrable} if

(i) each of them is integrable;

(ii) any  sum $\mathcal{D}_i  + \mathcal{D}_j +\cdots +\mathcal{D}_s$ (where all indices $i,j,\ldots, s$ are different) is also integrable.
\end{definition}

In 1955, J. Haantjes proved a crucial result:

\begin{theorem}[\cite{Haa}]\label{th:Haan}
Let $\boldsymbol{L}: TM \to TM$ be an   operator field;  assume that the rank of each generalized eigen-distribution $\mathcal{D}_i$, $i=1,\ldots, s$ is independent of $\boldsymbol{x}\in M$.
The vanishing of the Haantjes torsion
\begin{equation}
 \label{eq:HaaNullTM}
\mathcal{H}_{\boldsymbol{L}}(X, Y)= \mathbf{0} \qquad\qquad\qquad \forall ~ X, Y \in TM
\end{equation}
is a  sufficient condition to ensure the mutual integrability of the  generalized eigen-distributions $\{\mathcal{D}_1,\ldots, \mathcal{D}_s\}$.
In addition, if $\boldsymbol{L}$ is  semisimple, condition \eqref{eq:HaaNullTM} is also necessary.
\end{theorem}

Consequently, under the previous assumptions, one can  select  local coordinate charts in  which $\boldsymbol{L}$ takes a  block-diagonal form. An equivalent statement can be formulated in terms of the existence of integrable generalized eigen-frames  of $\boldsymbol{L}$.

\begin{proposition}\label{cor:Hframe}
The vanishing of the Haantjes torsion of an operator field $\boldsymbol{L}$ is a  sufficient condition to ensure that $\boldsymbol{L}$ admits an equivalence class  of integrable generalized eigen-frames, where $\boldsymbol{L}$ takes a block-diagonal form. Furthermore, if $\boldsymbol{L}$ is semisimple, the condition is also necessary and $\boldsymbol{L}$ takes a diagonal form; if $\boldsymbol{L}$ is \emph{simple} each  of its proper eigen-frames is integrable.
\end{proposition}

\subsection{Haantjes algebras}

The notion of Haantjes algebra, introduced and discussed in \cite{TT2021}, is a crucial piece of the geometric construction we wish to propose for the analysis of separable systems. 
\begin{definition}\label{def:HM}
A Haantjes algebra of rank $m$ is a pair    $(M, \mathscr{H})$ with the following properties:
\begin{itemize}
\item
$M$ is a differentiable manifold of dimension $\mathrm{n}$;
\item
$ \mathscr{H}$ is a set of Haantjes  operators $\boldsymbol{K}:TM\rightarrow TM$   that  generate
\begin{itemize}
\item
a free module of rank $\mathrm{m}$   over the ring of smooth functions on $M$:
\begin{equation}
\label{eq:Hmod}
\mathcal{H}_{\bigl( f\boldsymbol{K}_{1} +
                             g\boldsymbol{K}_2\bigr)}(X,Y)= \mathbf{0}
 \ , \qquad\forall\, X, Y \in TM \ , \quad \, f,g \in C^\infty(M)\  ,\quad \forall ~\boldsymbol{K}_1,\boldsymbol{K}_2 \in  \mathscr{H};
\end{equation}
  \item
a ring  w.r.t. the composition operation
\begin{equation}
 \label{eq:Hring}
\mathcal{H}_{\bigl(\boldsymbol{K}_1 \, \boldsymbol{K}_2\bigr)}(X,Y)=\mathbf{0} \ , \qquad
\forall\, \boldsymbol{K}_1,\boldsymbol{K}_2\in  \mathscr{H} , \quad\forall\, X, Y \in TM\ .
\end{equation}
\end{itemize}
\end{itemize}
If
\begin{equation}
\boldsymbol{K}_1\,\boldsymbol{K}_2=\boldsymbol{K}_2\,\boldsymbol{K}_1 \ , \quad\qquad\ \boldsymbol{K}_1,\boldsymbol{K}_2 \in  \mathscr{H}\ ,
\end{equation}
the  algebra $(M, \mathscr{H})$ will be said to be an Abelian Haantjes algebra. Moreover, if   the identity operator $\boldsymbol{I}\in \mathscr{H}$, then $(M, \mathscr{H})$ will be said to be a Haantjes algebra with identity.
\end{definition}
In other words, the set $\mathscr{H}$ can be regarded as an associative algebra of Haantjes operators. Observe that if $\boldsymbol{K} \in \mathscr{H}$, then the powers $\boldsymbol{K}^i \in \mathscr{H}  ~ \forall\, i\in \mathbb{N} \backslash \{0\}$.

Haantjes algebras possess several important properties. Among them, we recall that for a given Abelian Haantjes algebra $\mathscr{H}$ there exists associated a set of  coordinates, called \textit{Haantjes coordinates},   by means of which  all $\boldsymbol{K}\in \mathscr{H}$ can  be  written simultaneously in a  block-diagonal form.  
In particular, if $\mathscr{H}$ is also semisimple, on each set of Haantjes coordinates all $\boldsymbol{K}\in \mathscr{H}$ can  be written simultaneously  in a  diagonal form \cite{TT2021}.

\subsection{Haantjes chains}
The notion of Haantjes chains, which generalizes that of Lenard-Magri chains \cite{Magri78,MagriLE}  has been proposed in \cite{TT2016prepr}. In the forthcoming analysis, Haantjes chains will enable us to build a bridge between the Haantjes geometry and the theory of separable systems. Other generalizations have also been proposed in the literature of the last decades \cite{MT,MTPLA,MTRomp,FMT,FP}.

%

\begin{definition}
 Let $( M,\mathscr{H})$ be a Haantjes algebra of rank $\mathrm{m}$. A  function $H\in C^{\infty}(M)$ is said to generate a Haantjes chain of 1-forms of length $\mathrm{m}$ if  there exist
a distinguished basis  $\{\boldsymbol{K}_1,\ldots, \boldsymbol{K}_m\}$ of $\mathscr{H}$
 such that
\begin{equation} \label{eq:MHchain}
\rd (\boldsymbol{K}^T_\alpha \,\rd H )=\boldsymbol{0}  \ ,  \quad\qquad \alpha=1,\ldots ,m 
\end{equation}
where $\boldsymbol{K}^{T}_{\alpha}: T^{*}M \to T^{*}M$ is the transposed operator of $\boldsymbol{K}_{\alpha}$ . The (locally) exact 1-forms $\rd H_i$ such that
$$
\rd H_\alpha=\boldsymbol{K}^T_\alpha \,\rd H \ 
$$
(assumed to be linearly independent) are called the elements of the Haantjes chain of length $\mathrm{m}$ generated by $H$ and the functions $H_\alpha\in C^{\infty}(M)$ are their potential functions.
\end{definition}

Given a basis $ \{ \boldsymbol{K}_1,  \boldsymbol{K}_2,\ldots,\boldsymbol{K}_{m}\} $ of $\mathscr{H}$, let us denote by
\begin{equation} \label{eq:codKH}
\mathcal{D}_H^\circ:= \langle \boldsymbol{K}_1^T \rd H,  \boldsymbol{K}_2^T \rd H,\ldots,\boldsymbol{K}_{m}^T\,\rd H \rangle
\end{equation}
the co-distribution generated by a function $H$, and by $\mathcal{D}_H$ the distribution  of the vector fields annihilated by them (of rank  $(n-m)$). A result proved in \cite{TT2016prepr} states that the function $H$  generates a Haantjes chain $\eqref{eq:MHchain}$
if and only if $\mathcal{D}^\circ_H$ (or equivalently $\mathcal{D}_H$) is Frobenius-integrable.

Now, we shall briefly review the theory of $\omega \mathscr{H}$ or symplectic-Haantjes manifolds, firstly introduced in \cite{TT2016prepr}. They offer a natural theoretical framework for the formulation of the theory of Hamiltonian integrable systems. 

\subsection{$\omega \mathscr{H}$ manifolds}

\begin{definition}\label{def:oHman}
A symplectic--Haantjes (or $\omega \mathscr{H}$) manifold  of class $\mathrm{m}$ is a triple $( M,\omega,\mathscr{H})$ which satisfies the following properties:
\begin{itemize}
\item[i)]
$(M,\omega)$  is a   symplectic  manifold of dimension $\mathrm{ 2 n}$;
\item[ii)]
$\mathscr{H}$ is a Haantjes algebra of rank $\mathrm{m}$;
\item[iii)]
$(\omega,\mathscr{H})$ are algebraically compatible, that is
$$
\omega(X,\boldsymbol{K} Y)=\omega(\boldsymbol{K} X,Y)  \qquad \forall \boldsymbol{K} \in \mathscr{H}\ ,
$$
or equivalently
\begin{equation}\label{eq:compOmH}
\boldsymbol{\Omega}\, \boldsymbol{K} =\boldsymbol{K}^T\boldsymbol{\Omega} ,\qquad \ \forall \boldsymbol{K} \in \mathscr{H}\ .
\end{equation}
\end{itemize}
\noi
Hereafter $\boldsymbol{\Omega}:=\omega ^\flat:TM\rightarrow T^*M$ denotes the  fiber bundles isomorphism defined by
$$
\omega(X,Y)=\langle\boldsymbol{\Omega} X,Y \rangle\qquad\qquad\forall X, Y \in TM\ ,
$$
and the map $\boldsymbol{P}:=\boldsymbol{\Omega}^{-1}:T^*M \rightarrow TM$  is the Poisson bivector induced by  the symplectic structure $\omega$.
\par
If the identity operator $\boldsymbol{I}$ belongs to $\mathscr{H}$, then $( M,\omega,\mathscr{H})$ will be said to be an $\omega \mathscr{H}$ manifold with identity. If $\mathscr{H}$ is an Abelian Haantjes algebra, we shall say that the resulting $\omega \mathscr{H}$ manifold is Abelian.
\end{definition}

\begin{definition}
An $\omega \mathscr{H}$ manifold $(M, \omega, \mathscr{H})$ is semisimple if $ \mathscr{H}$ is a semisimple Haantjes algebra.

\end{definition}
Observe that a simple realization of the notion of Haantjes algebra is given in a Darboux chart
$\{ \boldsymbol{x}=(q^1,\ldots q^n, p_1,\ldots,p_n)\}$ by
\begin{equation}\label{eq:Hdiag}
\boldsymbol{K}_\alpha=\sum _{i=1}^{n} l_{i }^{(\alpha)}(\boldsymbol{x})
\Big(\frac{\partial}{\partial q^i}\otimes \rd q^i + \frac{\partial}{\partial p_i}\otimes \rd p_i\Big ), \qquad\qquad \alpha=1,\ldots, m\ ,
 \end{equation}
where $l_{i}^{(\alpha)}= \lambda_i^{(\alpha)}(\boldsymbol{x})=\lambda_{n+i}^{(\alpha)}(\boldsymbol{x})$, $i=1,\ldots,n$.

We proved in \cite{TT2016prepr} that there exists a spectral decomposition of the tangent spaces $T_{\bs{x}}M= \bigoplus_{i=1}^{s}  \mathcal{D}_i(\bs{x})$  realized in terms of
 (generalized) eigenspaces $\mathcal{D}_i(\bs{x})$ of \textit{even rank}. Consequently, the number of the distinct eigenvalues of any Haantjes operator $\boldsymbol{K}$ of an $\omega\mathscr{H}$ structure is not greater than $n$. In particular, if the number of  distinct eigenvalues of an operator $\boldsymbol{K}\in \mathscr{H}$ is exactly $n$, the  operator will be said to be\textit{ maximal}. This is equivalent to require that the minimal polynomial of $\bs{K}$,
$m_{\boldsymbol{K}}(\boldsymbol{x},\lambda)=
\prod_{i=1}^n \Big(\lambda - l_i(\boldsymbol{x})\Big)^{\rho_i}$,
 has degree $m=n$.

Several other interesting results can be stated in the $\omega \mathscr{H}$ geometry. 
In particular, given a   $\omega \mathscr{H}$ manifold, the distributions $ \mathcal{D}_j$, $j=1,\ldots,s$ of each $\boldsymbol{K} \in \mathscr{H}$ are integrable and of even rank.  Besides, their  integral leaves are symplectic submanifolds of $M$ and are symplectically orthogonal to each other, namely $\omega (\mathcal{D}_j,\mathcal{D}_k)  =  \boldsymbol{0}$,  $ j\neq k$.

\subsection{Darboux--Haantjes coordinates for $\omega\mathscr{H}$ manifolds}

Assume that $(M,\omega,\mathscr{H})$ is an Abelian $\omega\mathscr{H}$ manifold of class $m$. Then, there exist local  charts in $U\subset M$  which are Darboux coordinates for $\omega$; besides, all of the Haantjes operators take simultaneously a block-diagonal form \cite{TT2016prepr}. Due to their twofold role, they will be called \textit{Darboux--Haantjes} (DH) coordinates.

\begin{corollary}
Given a semisimple Abelian $\omega\mathscr{H}$ manifold $(M, \omega ,\mathscr{H})$, on a set of  Darboux--Haantjes coordinates  each $\boldsymbol{K} \in \mathscr{H}$ takes the diagonal form \eqref{eq:Hdiag}.
\end{corollary}

The relevance of Haantjes chains in the theory of  $\omega\mathscr{H}$ manifolds is due to the following
\begin{lemma} \label{lm:MHchainInv}
 Let $(M,\omega,\mathscr{H})$ be an Abelian $\omega\mathscr{H}$ manifold. Then the potential functions $H_\alpha\in C^{\infty}(M)$ of the Haantjes chain  generated by a distinguished function $H\in C^{\infty}(M)$ are in involution with $H$ and among each other, w.r.t. the Poisson bracket defined by the Poisson operator
 $\boldsymbol{P}=\boldsymbol{\Omega}^{-1}$.
\end{lemma}

\subsection{Cyclic $\omega \mathscr{H}$  manifolds }  \label{sec:cOmH}
Cyclic Haantjes algebras are a simple and very interesting instance of Haantjes algebras: they are cyclically generated by  a suitable Haantjes operator $\boldsymbol{L}$, according the following 
\begin{definition} \label{def:CHa}
Let  $(M,  \mathscr{H})$ be  a  rank $\mathrm{m}$ Abelian Haantjes algebra. An  operator $\boldsymbol{L}$ whose minimal polynomial is of degree $\mathrm{h}\geq \mathrm{m}$  will be said to be  a generator of  $\mathscr{H}$  if
\begin{equation*}
 \mathscr{H}\subseteq\mathcal{L}(\boldsymbol{L})=\left < \boldsymbol{I} , \boldsymbol{L},\boldsymbol{L}^2, \ldots, \boldsymbol{L}^{h} \right > \ . 
\end{equation*}
\end{definition}

 In other words, all the operators of the algebra are of the form $\boldsymbol{K}_\alpha = p_\alpha (\boldsymbol{L})$, where $p_\alpha(\boldsymbol{x},\lambda)$ is a suitable polynomial in $\lambda$ of degree not greater than $(h-1)$ with coefficients in $C^{\infty}(M)$.   
A \emph{cyclic} $\omega  \mathscr{H}$ manifold is an $\omega  \mathscr{H}$ manifold endowed with a cyclic Haantjes algebra.
\par
\begin{lemma}\label{lm:gc}
 Let $(M,\omega,\mathscr{H})$ be an Abelian $\omega\mathscr{H}$ manifold of rank m. An operator 
 $\boldsymbol{L}$ belonging to $\mathscr{H}$ is a generator of  $\mathscr{H}$ if and only if  its minimal polynomial has   degree $h=m$ whenever $\boldsymbol{I} \in  \mathscr{H}$, and $h=m+1$ otherwise.
\end{lemma}
\begin{proof}
Firstly, let us consider the case $\boldsymbol{I} \in  \mathscr{H}$.  As   $\boldsymbol{L} \in  \mathscr{H}$,     
$\mathcal{L}(\boldsymbol{L})\subseteq \mathscr{H}$.   Let us assume that $\boldsymbol{L}$ is a generator of $\mathscr{H}$.  Then, by definition,  $ \mathscr{H}\subseteq\mathcal{L}(\boldsymbol{L})$ and the degree of the minimal polynomial of $\bs{L}$ is not smaller than $m$. Thus,      $\mathcal{L}(\boldsymbol{L})= \mathscr{H}$ and its   rank is equal to $m$. Therefore, the minimal polynomial of the  operator  $\bs{L}$  has degree $m$. Conversely, let us suppose that $\mathscr{H}$ contains an operator $\bs{L}$ whose   minimal polynomial  is of degree $m$. Then the rank of $\mathcal{L}(\boldsymbol{L})$ is $m$; therefore,  $\mathcal{L}(\boldsymbol{L})= \mathscr{H}$, so $\bs{L}$ is a generator of   $ \mathscr{H}$.
\par
In the case $\boldsymbol{I} \notin  \mathscr{H}$, consider the extended Haantjes algebra of rank $m+1$ obtained by adding the identity operator to $\mathscr{H}$. Then,    we are led to the previous case.    Thus, $\boldsymbol{L} \in \mathscr{H}$ is a generator of the extended algebra (and   therefore of $\mathscr{H}$) if and only if its minimal polynomial has degree $m+1$.

\end{proof}

\subsection{$\omega N$ and $\omega \mathscr{H}$  manifolds}

An interesting family of cyclic $\omega \mathscr{H}$ manifolds is represented by     $\omega N$ manifolds  \cite{MM1984,Mnoi}. In that context, cyclic Haantjes chains convert into Nijenhuis chains, as in \cite{FMT}, or  generalized Lenard chains as in  \cite{TT2012,TGalli12}. In the following, we shall discuss the relation between these two geometric structures. Precisely, we shall prove that in the Abelian, semisimple case, there is a one-to-one correspondence between $\omega\mathscr{H}$ manifolds and equivalence classes of $\omega N$ manifolds. This correspondence does not hold in the more general, non-semisimple case.

Let us recall that an $\omega N$ manifold $(M,\omega, \bs{N})$ is a symplectic manifold endowed with a  Nijenhuis operator $\bs{N}$ compatible with $\omega$, that is 
\begin{eqnarray}
\label{eq:OmNa}
 \bs{\Omega N}&=&\bs{N}^T \bs{\Omega}, \\
 \label{eq:OmNd}
 \rd ( \bs{\Omega N} ) &=&\bs{0}
\end{eqnarray}
where $\bs{\Omega}: TM \rightarrow T^*M$ is the linear map defined by $\bs{\Omega}:=\omega^\flat$.
Condition \eqref{eq:OmNa} is equivalent to the fact that the composed linear map $\bs{\Omega N}$ is also skew-symmetric, therefore the tensor field
$\omega_1$ induced  by the linear map $\omega_1 ^\flat:=\bs{\Omega N} $, is a $2$-form.  Thus, we  can  compute its  exterior derivative, which is imposed to vanish by condition \eqref{eq:OmNd}.
\par

First, let us consider an $\omega N$ manifold, and assume that its Nijenhuis operator $\boldsymbol{N}$ admits a minimal polynomial  of degree $m$. Then, the manifold $M$ has associated a cyclic $ \omega \mathscr{H}$ structure, given by
\begin{equation} \label{def:sOmH}
(M, \omega,  \boldsymbol{K}_1= \boldsymbol{I},  \boldsymbol{K}_2= \boldsymbol{N},\ldots,   , \boldsymbol{K}_{m}= \boldsymbol{N}^{m-1}) \ ,
\end{equation}
with a Haantjes algebra of rank  $m\leq dim(M)$.
In fact, each Nijenhuis operator $ \boldsymbol{N}$ is also a Haantjes operator; therefore, it generates the  cyclic Haantjes algebra $\mathcal{L}(\boldsymbol{N})$. Furthermore, the algebraic compatibility condition \eqref{eq:OmNa} assures that for all Haantjes operators
\begin{equation} \label{eq:KOmN}
\boldsymbol{K}=p_{\boldsymbol{K} }(\boldsymbol{x},\boldsymbol{N})=\sum_{i =0} ^{m-1}  a_i(\boldsymbol{x})\,\boldsymbol{N}^i ,
\end{equation}
condition iii) of Definition \ref{def:oHman} is fulfilled.
 \par

\vspace{2mm}

Conversely, let us construct an $\omega N$ structure starting from an Abelian, semisimple $\omega\mathscr{H}$ manifold. As a consequence of Propositions 37 and 38 in \cite{TT2016prepr},  Abelian semisimple $\omega\mathscr{H}$ manifolds are always cyclic ones. Besides, their generator can be chosen to be a Nijenhuis operator. 
In this case, DH coordinates coincide with DN coordinates and one  can  take, as usual in the $\omega N$ theory, the  eigenvalue fields $(\lambda_1, \ldots \lambda_n)$ of a Nijenhuis generator of the Haantjes algebra, assumed to be functionally independent,  as half of the set of  DH coordinates. Regarding the other half, we propose in Sec. \ref{sec:Procedure} a simple method for  computing them, based on the analysis developed in \cite{TT2016prepr} about the characteristic web of an $\omega\mathscr{H}$ manifold. 

Let us recall that the spectral decomposition of the tangent spaces associated to a Haantjes generator $\boldsymbol{L}$ of a semisimple Abelian Haantjes algebra is given by
\begin{equation} \label{eq:SpecD}
T_{\boldsymbol{x}}M= \bigoplus_{i=1}^n \mathcal{D}_i(\boldsymbol{x})=
\mathcal{D}_i (\boldsymbol{x}) \oplus \mathcal{E}_i(\boldsymbol{x})
\end{equation}
where 
\begin{equation} 
\mathcal{D}_i= \ker \left(\boldsymbol{L}-l_i\boldsymbol{I}\right), \qquad
 \mathcal{E}_i= Im\left(\boldsymbol{L}-l_i\boldsymbol{I}\right)= \bigoplus_{j=1, j\neq i}^n \mathcal{D}_j \qquad    i=1,\ldots,n.
\end{equation}
Correspondingly,  
\begin{equation}\label{eq:E0}
\mathcal{E}^\circ_i=  \ker \left(\boldsymbol{L}^T-l_i\boldsymbol{I}\right )
\end{equation}
 and 
the cotangent spaces decompose as 
\begin{equation}
T^*_{\boldsymbol{x}}M= \bigoplus_{i=1}^n \mathcal{E}^\circ_i(\boldsymbol{x}) \ .
\end{equation}
The ranks of the eigendistributions $\mathcal{D}_i$,  of the characteristic  distributions  $\mathcal{E}_i$ and  of their annihilators $ \mathcal{E}^\circ_i$, $i=1,\ldots,n$ are  $2$, $2n-2$, $2$, respectively . 
\par
As a further step, we can prove that given an Abelian semisimple $\omega \mathscr{H}$ manifold of class $n$, there exists an $\omega N$ manifold associated. To this aim, for the sake of clarity, we specialize Proposition 38, proved in  \cite{TT2016prepr}, to the case  $m=n$.
\begin{proposition} \label{th:HgDH}
Let  $(M,\omega,\mathscr{H})$ be an Abelian $2n$-dimensional semisimple $\omega\mathscr{H}$ manifold of class $n$. Let us consider the spectral decomposition \eqref{eq:SpecD}  and a Darboux-Haantjes chart
 $\{ U,(q^i,p_i) \}$, $i=1,\ldots , n$,  adapted to the decomposition \eqref{eq:SpecD}, namely
\begin{equation}\label{eq:HchartK}
\mathcal{D}_{i}=\left\langle\frac{\partial}{\partial q^i},\frac{\partial}{\partial p_i}\right\rangle \ .
\end{equation}
Then, each operator defined by
 \begin{equation}\label{eq:Lgen}
\boldsymbol{L}=\sum_{i=1}^n \lambda_{i}(\boldsymbol{q,p})\bigg(\frac{\partial}{\partial q^{i}}\otimes \rd q^{i}+\frac{\partial}{\partial p_{i }}\otimes \rd p_{i }\bigg)
 \end{equation}
 is a generator of $\mathscr{H}$, provided that $\{\lambda_{1}(\boldsymbol{q, p}), \ldots , \lambda_{n}(\boldsymbol{q, p})\}$ are arbitrary,
pointwise distinct smooth functions. Therefore, every operator $\boldsymbol{K}\in \mathscr{H}$ can be written in the form
 \begin{equation}\label{eq:KLagr}
 \boldsymbol{K}=p_K(\bs{x},\bs{L})=\sum _{i=1}^n l_{i } \frac{\Pi_{j\neq i}(\boldsymbol{L}-\lambda_j \boldsymbol{I})}{\Pi_{j\neq i}(\lambda_i-\lambda_j )} \ , 
 \end{equation}
where $l_i=l_i(\boldsymbol{q, p})$ are the eigenvalue fields of $\boldsymbol{K}$.
In particular, if
\begin{equation}\label{eq:64}
\lambda_i(\boldsymbol{q,p})=\lambda_i(q^{i},p_{i}) \qquad\qquad i=1,\ldots, n \ ,
\end{equation}
each generator $\boldsymbol{L}$ is a Nijenhuis operator. Each of these  operators endows the manifold M with an  $\omega N$ structure, for which the DH chart $\{ U,(q^i,p_i) \}$, $i=1,\ldots , n$ is a DN chart.
\par
\end{proposition}

\begin{proof}
 The first part of the Theorem has been proved in \cite{TT2016prepr}.
 \par
  Let us prove now that both conditions \eqref{eq:OmNa} and \eqref{eq:OmNd} are fulfilled by a Nijenhuis generator, that is a generator \eqref{eq:Lgen} of the algebra $\mathscr{H}$ which satisfies condition \eqref{eq:64}.  We denote  such Nijenhuis generator by $\bs{N}$.
Condition \eqref{eq:OmNa} is satisfied as  $\bs{N}$  takes a diagonal form in the DH coordinates $(\boldsymbol{q,p})$.
Concerning condition \eqref{eq:OmNd}, it suffices to observe that in any set of  DH coordinates, the $2$-form $\bs{\Omega N}$ takes the local expression
\begin{equation}
\bs{\Omega N}=
\sum_{i=1}^n \lambda_i(q^{i},p_{i})\, \rd p_{i} \wedge   \rd q^{i} \ .
 \end{equation}
Therefore, its exterior derivative vanishes, as can be computed  by means of a direct calculation.
\end{proof}
We can now make precise the meaning of the correspondence between $\omega \mathscr{H}$ manifolds and equivalence classes of $\omega N$ manifolds.

Given a semisimple $\omega \mathscr{H}$ manifold, it admits infinitely many Nijenhuis generators. However, they are all related. Indeed, the corresponding $\omega N$ structures, although \textit{a priori} different, possess common DN coordinates which in turn are also DH coordinates for the algebra $\omega \mathscr{H}$. We shall say that two semisimple $\omega N$ structures are equivalent if they possess common DN coordinates. It is easy to show that this property  defines an equivalence relation of $\omega N$ manifolds. Thus, given an $\omega \mathscr{H}$ manifold, there exists associated an equivalence class of $\omega N$ manifolds. The vice versa is obvious, since all Nijenhuis operators of the class by construction generate the same cyclic $\omega \mathscr{H}$ structure.

More generally, given an equivalence class of Abelian, semisimple $\omega N$ manifolds, one can associate a unique $\omega \mathscr{H}$ manifold. 
Indeed, the powers of the Nijenhuis operators of the class all diagonalize in the same DN chart. Therefore, from the $C^{\infty}(M)$ modules generated by them one can define an algebra of diagonal Haantjes operators compatible with $\omega$, i.e., an $\omega \mathscr{H}$ manifold. Obviously, the DN chart is a DH chart for the $\omega \mathscr{H}$ manifold.


\subsection{A procedure to construct DH coordinates} \label{sec:Procedure} 

We wish to envisage a strategy for the explicit construction of families of Darboux-Haantjes coordinates.

We shall assume that the  distributions  $\mathcal{D}_j$, $\mathcal{E}_i$, $\mathcal{E}_i^\circ$  have constant rank. Thus,  they are  integrable by virtue of the Haantjes theorem and fulfill the properties summed up  in the following 

\begin{proposition}\label{pr:FD}
Given a   semisimple $\omega \mathscr{H}$ manifold of class $n$, the eigen-distributions $ \mathcal{D}_j$ of each generator $\boldsymbol{L}$ are integrable and of  rank $2$.  Their  integral leaves are two-dimensional symplectic submanifolds of $M$ and are symplectically orthogonal to each other:
\begin{eqnarray}
\mathbf{\Omega}(\mathcal{D}_j)&=&\mathcal{E}_j^\circ \ \Leftrightarrow
\mathcal{D}_j = \boldsymbol{P}(\mathcal{E}_j^\circ)= \mathcal{E}_j^\perp  \ , \label{eq:PE0} \\
\label{eq:OmE}
\mathbf{\Omega}(\mathcal{E}_j)&=&\mathcal{D}_j^\circ\ \Leftrightarrow
\mathcal{E}_j = \boldsymbol{P}(\mathcal{D}_j^\circ)= \mathcal{D}_j^\perp \ , \\  
\omega({\mathcal{D}_j,\mathcal{D}_j})&=&symplectic,    \label{eq:Dsym} \\
\omega (\mathcal{D}_j,\mathcal{D}_k) & = & \boldsymbol{0}  \label{eq:Dsyo} \qquad\qquad j\neq k \ .
\end{eqnarray}
Here $\mathcal{E}_j^\perp$ and $\mathcal{D}_j^\perp$ are the symplectic orthogonal distributions of $\mathcal{E}_j$ and $\mathcal{D}_j$, respectively.
\end{proposition}
This result is a consequence of  Propositions 23 and 24 of \cite{TT2016prepr}.
\par
Under the same assumptions of Proposition \ref{pr:FD},  DH coordinates  can be determined  as pairs of  characteristic functions of the Haantjes web of the generator $\boldsymbol{L}$, namely functions $(x_i,y_i)$ which are constant on the characteristic eigendistributions $\mathcal{E}_i$. These functions are the same as the potential functions of the exact eigen-forms of $\boldsymbol{L}^T$ (due to eq. \eqref{eq:E0}). Also,  they fulfill the involution relations
\begin{eqnarray}
\{ x_i ,x_j\}&=0 \ , \quad \{ y_i ,y_j\}=0 \ , \qquad &i=1,\ldots,n \\
\{x_i,y_j\}&=0 \qquad\qquad \qquad\qquad\qquad &i\neq j=1,\ldots,n
\end{eqnarray}
as a consequence of properties  \eqref{eq:PE0} and \eqref{eq:Dsyo}. Therefore, our problem is to find pairs of characteristic functions $(x_i,y_i)$ that, in addition, are canonically conjugated:
$$
\{x_i,y_i\}=1, \qquad\qquad i=1,\ldots,n \ .
$$
They represent Darboux coordinates for each of bidimensional symplectic distributions $\mathcal{E}^\circ_i$.  We present now an effective procedure for the determination of these coordinates, which can be summed up in the following steps.

\par
1) We determine a basis of 1-forms $\{ \gamma_i,  \delta_i\}$ for  the characteristic co-distributions  $\mathcal{E}_i^\circ$  of the Haantjes generator $\bs{L}$.  \\
2)   For each $i=1,\ldots,n$, we search for an exact 1-form $\alpha_i=\rd x_i \in \mathcal{E}_i^\circ$,  by representing it as  
\begin{equation} \label{eq:momenti}
\alpha_i = f_i \,\gamma_i +g_i\, \delta_i\ ,
\end{equation}
and by requiring that    $\rd \alpha_i=0$. This step can be managed with a suitable ansatz about  the form of the two integrating factors $f_i$ and $g_i$.
\\
3)  We find the potential function $x_i$ of the exact 1-form $\alpha_i$. The potential functions so determined are automatically in involution w.r.t. the Poisson bracket induced by $\omega$.
\\
4)  We search for another exact 1-form $\beta_i=\rd  y_i \in \mathcal{E}_i^\circ$, linearly independent of  $\rd x_i$. To this aim,  we can  represent $\beta_i$ as  
\begin{equation} \label{eq:momenti}
\beta_i =h_i \, \rd x_i+ r_i \,\gamma_i \ ,
\end{equation}
where we have assumed, without loss of generality, that $\gamma_i \in  \mathcal{E}_i^\circ$ is linearly independent of $ \rd x_i$. \\
5) By requiring that the potential function $y_i$  is canonically conjugated with $x_i$, that is 
$$
1=\{ x_i, y_i \}=\langle \rd x_i, \boldsymbol{P} \,\rd y_i \rangle = 
h_i\ \cancel{\langle \rd x_i, \boldsymbol{P} \,\rd x_i \rangle}+
r_i\langle \rd x_i, \boldsymbol{P} \, \gamma_i \rangle \ ,
$$
we get the normalizing factor 
$$
r_i= \frac{1}{\langle \rd x_i, \boldsymbol{P} \gamma_i \rangle}.
$$
6) Substituting $r_i$ into eq. \eqref{eq:momenti} and  imposing that the 1-form   $\beta_i$ is  a closed 1-form, we find the last  normalizing factor $h_i$. \\
7) Finally, we find the potential function $y_i$ of the exact 1-form $\beta_i$.
\par \vs
The previous  task  is much easier if one already knows half  of the DH coordinates, say $\{x_i\}_{1 \leq i\leq n}$. In this case, one skips the second and third step of the previous  procedure. For instance, this occurs if one has   a Haantjes generator $\boldsymbol{L}$ which is also Nijenhuis operator. In fact, in this case, the eigenvalues fields of $\boldsymbol{L}$ are just  characteristic functions of the Haantjes web. Then, once the eigenvalues of $\boldsymbol{L}$ have been computed by algebraic methods, the conjugated momenta   $\{y_i\}_{1 \leq i\leq n}$ can be determined by means of the steps 1), 4), 5), 6), 7). We shall 
apply this shorter procedure in Section \ref{sec:Drach}.
\vs

\section{Separation of variables in $\omega \mathscr{H}$ manifolds} \label{sec:separation}

\subsection{Main theorem}
\noi The next theorem represents our main result concerning the existence of separation variables in the theory  of $\omega \mathscr{H}$ manifolds. 
\begin{theorem} [Jacobi-Haantjes] \label{th:SoVgLc}
Let $M$ be an Abelian  semisimple $\omega \mathscr{H}$ manifold of class $n$ and $\{H_1,H_2,\ldots,H_n\}$ be a set of $C^{\infty}(M)$ functions  belonging to a    Haantjes chain generated by a function $H\in C^{\infty}(M)$ via a basis of operators $\{\boldsymbol{K}_1,\ldots,\boldsymbol{K}_{n}\}\in \mathscr{H}$. 
Then, each set $(\boldsymbol{q},\boldsymbol{p})$  of {\rm DH} coordinates
provides us with separation variables for the Hamilton--Jacobi equation associated with each function $H_j$.
 \par
 Conversely, if $M$ is a symplectic manifold and $\{H_1,H_2,\ldots,H_n\}$ are  $n$ independent,  $C^{\infty}(M)$ functions separable in a set of Darboux coordinates $(\boldsymbol{q},\boldsymbol{p})$,
 then they belong to the Haantjes  chain generated by the operators 
 \begin{equation} \label{eq:LSoV}
\boldsymbol{K}_\alpha=\sum _{i=1}^n \frac{\frac{\partial H_{\alpha}}{\partial p_i}}{ \frac{\partial H}{\partial p_i}}\bigg (\frac{\partial}{\partial q^i}\otimes \rd q_i +\frac{\partial}{\partial p_i}\otimes \rd p_i \bigg )\qquad \alpha=1,\ldots,n \ ,
\end{equation}
where $H$ is any of the functions $\{H_1, \ldots, H_n\}$, with $\frac{\partial H}{\partial p_i}\neq 0$, $i=1,\ldots,n$. These operators generate a semisimple $\omega \mathscr{H}$ structure on $M$.
\end{theorem}
\begin{proof}
Theorem 25 of \cite{TT2016prepr} guarantees the existence of  sets of DH coordinates for a semisimple $\omega \mathscr{H}$ manifold. Therefore, it suffices to show that the functions $H_j$ in
such coordinates are in separable involution, according to eq. \eqref{eq:SI}.
 To this aim, let us note that, due  to the
diagonal form of  $\boldsymbol{K}_\alpha^T$ in a  DH local chart, the relations
\begin{eqnarray} \label{eq:dHDNgLc}
\frac{\partial H_j}{\partial q^k}&=&l^{(j)}_k\frac{\partial H}{\partial
q^k}\ ,\qquad\qquad  j,k=1,\ldots, n \ , \\
\frac{\partial H_j}{\partial p_k}&=&l^{(j)}_k\frac{\partial H}{\partial p_k} \ ,\qquad\qquad  j,k=1,\ldots, n,
\end{eqnarray}
hold. Here $l^{(j)}_k$ denotes the $k$-th eigenvalue of the  Haantjes operator $\boldsymbol{K}_{j}^T$.
Therefore,
$$\{H_i,H_j\}_{\vert k}=l^{(i)}_k\,\frac{\partial H}{\partial
q^k}\,l^{(j)}_k\,\frac{\partial H}{\partial p_k}
-l^{(j)}_k\,\frac{\partial H}{\partial
q^k}\,l^{(i)}_k\,\frac{\partial H}{\partial p_k} = 0 \ .$$

%
In order to prove the converse statement, without loss of generality we can assume that $ { \frac{\partial H}{\partial p_i}}\neq 0 $ \textit{for} $ i=1,\ldots,n$.
 The operators \eqref{eq:LSoV}, being diagonal in the separated coordinates, are Haantjes operators. Also, they commute with each other and generate an Abelian, semisimple Haantjes algebra $\mathscr{H}$. The algebraic compatibility conditions \eqref{eq:compOmH} of the operators \eqref{eq:LSoV} with the symplectic form are equivalent to the further conditions
\beq
 l_{n+i}^{(\alpha)}=l_i^{(\alpha)} \qquad i=1,\ldots, n \ .
\eeq
Thus, the Haantjes operators \eqref{eq:LSoV} must possess at least double eigenvalues.

Finally, we impose that the integrals of motion $\{ H_1,H_2, \ldots,H_n \}$ form a Haantjes chain, which will be  generated by any of these functions, denoted by $H$. Since $\boldsymbol{K}_\alpha$ $(\al=1,\ldots,n)$ is diagonal in the $(\boldsymbol{q},\boldsymbol{p})$ variables,  such conditions are equivalent, for each $\alpha$, to the overdetermined system of $2n$ algebraic equations in the $n$ functions $l_i^{(\alpha)}$
\begin{eqnarray}
\label{eq:lq}
l^{(\alpha)}_i \frac{\partial H}{\partial q^i}&= & \frac{\partial H_{\alpha}}{\partial q^i} \ ,\\
\label{eq:lp}
l^{(\alpha)}_i \frac{\partial H}{\partial p_i}&= & \frac{\partial H_{\alpha}}{\partial p_i} \ ,
\end{eqnarray}
 $ i=1,\ldots,n$. In fact, the above  equations are compatible, because the Benenti conditions \eqref{eq:SI}  of separate involution ensure that
$$
\frac{\partial H}{\partial q^i}
\frac{\partial H_{\alpha}}{\partial p_i}=\frac{\partial H}{\partial p_i}
\frac{\partial H_{\alpha}}{\partial q^i} \ , \quad 1\le i \le n \ .
$$
Consequently, Eqs. \eqref{eq:lp}   provide us with the unique solution \eqref{eq:LSoV}.
\end{proof}
\begin{corollary}
If $M$ is a symplectic manifold and $\{H_1,H_2,\ldots,H_n\}$ are  $n$ independent,  $C^{\infty}(M)$ functions separable in a set of Darboux coordinates $(\boldsymbol{q},\boldsymbol{p})$, the Haantjes operators \eqref{eq:LSoV} can be written as 
\begin{equation}\label{eq:KLagr}
 \boldsymbol{K}_\alpha=\sum _{i=1}^n \frac{\frac{\partial H_{\alpha}}{\partial p_i}}{ \frac{\partial H}{\partial p_i}} \,
 \frac{\Pi_{j\neq i}(\boldsymbol{N}-\lambda_j \boldsymbol{I})}{\Pi_{j\neq i}(\lambda_i-\lambda_j )} 
 \ ,
 \end{equation}
 where $\bs{N}$ is the operator defined by
 \begin{equation}\label{eq:Ngen}
\boldsymbol{N}=\sum_{i=1}^n \lambda_{i}(q^i, p_i)\bigg(\frac{\partial}{\partial q^{i}}\otimes \rd q^{i}+\frac{\partial}{\partial p_{i }}\otimes \rd p_{i }\bigg)
 \end{equation}
and $\{\lambda_{i}(q^i, p_i)\}_{1\leq i\leq n}$ are arbitrary, pointwise distinct smooth functions on $M$. The operator $\boldsymbol{N}$ is a Nijenhuis operator compatible with $\omega$ and, therefore, a generator of a cyclic $\omega \mathscr{H}$ structure on $M$, as defined by  \eqref{def:sOmH}.
\end{corollary}
\begin{remark}
Given $n$ arbitrary smooth functions $\{ H_{1},\ldots, H_{n}\}$ on a $2n$-dimensional manifold $M$, it is always possible to determine $n$ diagonal Haantjes operators $\bs{K}_{i}$ which satisfy the chain equations $\bs{K}^{T}_i dH=dH_i$, where $H$ is any of the previous functions. However, if $M$ is a symplectic manifold, the compatibility condition \eqref{eq:compOmH} of the Haantjes operators $\bs{K}_i$ with the symplectic structure defined on $M$ imposes $n$ additional constraints on the eigenvalues of these operators. Thus, the systems of $2n$ algebraic equations \eqref{eq:lq} and \eqref{eq:lp} can be solved if and only if $n$ of these equations are automatically satisfied. This requirement is equivalent to the Benenti conditions   \eqref{eq:SI}. Consequently, given an integrable system, the existence of a Haantjes chain  with semisimple operators (and consequently of a $\omega\mathscr{H}$ structure), is not at all a trivial property. The case of the superintegrable Post-Winternitz system is illustrative of this aspect: no separation variables are known for this system, and the unique known $\omega \mathscr{H}$ structures are \textit{non-semisimple ones}. 
\end{remark}

\subsection{A general procedure} \label{ss:GP}

In order to determine the $\omega \mathscr{H}$ structures admitted by a separable system, we need to construct the Haantjes chains associated with it. Precisely, we wish to solve the chain equations
\begin{equation} \label{eq:MHchain2}
\rd (\boldsymbol{K}^T_\alpha \,\rd H )=\boldsymbol{0}  \quad\qquad \alpha=1,\ldots,m 
\end{equation}
for suitable Haantjes operators $\boldsymbol{K}^T_{\alpha}$. Generally speaking, these equations do not admit a unique solution. Notice that the operators $\boldsymbol{K}^T_{\alpha}$ we are interested in must be compatible with the symplectic form $\omega$ (see eq. \eqref{eq:compOmH}).
 The most general operator  $\bs{M}$ compatible with the symplectic form, in Darboux coordinates $(\bs{q}, \bs{p})$ reads

\begin{equation} \label{eq:gen}
\bs{M}=\left[\begin{array}{c|c}\bs{A}(\bs{q},\bs{p}) & \bs{B(\bs{q},\bs{p})} \\
\hline \bs{C}(\bs{q},\bs{p}) & \bs{A}^{T}(\bs{q},\bs{p})
\end{array}\right]\ , \quad \bs{B}+\bs{B}^{T}= \bs{0} \ , \quad \bs{C}+\bs{C}^{T}= \bs{0} \ ,
\end{equation}

where $\bs{A}$, $\bs{B}$, $\bs{C}$ are $n\times n$ matrices with coefficients smoothly depending on the Darboux coordinates. However, the operator \eqref{eq:gen} in general is not a Haantjes operator, unless some specific choices for its arbitrary functions are made. Therefore, our task is to determine solutions of eqs. \eqref{eq:MHchain2} in the class of Haantjes operators of the form \eqref{eq:gen}. For $n=2$, the operators of the family \eqref{eq:gen} can depend on up to 6 arbitrary functions; for $n=3$, up to $15$ arbitrary functions are admissible.

The procedure for the determination of the $\omega \mathscr{H}$ structure associated with an $n$-dimensional integrable Hamiltonian system can be summarized in the following steps.

1) Determine the operators $\boldsymbol{K}$ of the family \eqref{eq:gen} that solve eqs. \eqref{eq:MHchain2}. \\
2) Among the solutions found, choose the operators satisfying the vanishing  condition for their Haantjes torsion: 
\[
\mathcal{H}_{\boldsymbol{K}}(X,Y)= \bs{0}, \qquad \forall~X, Y \in TM  \ .
\]
3) Find all the semisimple Abelian $\omega \mathscr{H}$ structures admitted by the system and their generators.

4) Find the DH coordinates associated with each $\omega \mathscr{H}$ structure. 

\vs

This procedure is completely general. 

\vs

\noi An alternative strategy is to  use a suitable ansatz for the form of the operators $\bs{K}$ that are supposed to solve eqs. \eqref{eq:MHchain2}. Precisely, one can  start from a sub-family of operators depending on arbitrary functions,  which are both  compatible with the symplectic form  and have vanishing Haantjes torsion. In this way, the requirement of step (2) is already fulfilled. Then, one can try to fix the arbitrary functions available by solving eqs. \eqref{eq:MHchain2} (step 1). This approach is less general, since it presupposes the  determination \textit{a priori} of special subclasses of operators of the form \eqref{eq:gen} which are Haantjes's as well. This task, which in general is computationally nontrivial, is affordable in the case $n=2$; consequently, this procedure has been adopted, for instance, in Section \ref{sec:anisotropic}. 

%

\section{Separation of variables for a  Drach-Holt type system} \label{sec:Drach}


The approach proposed in this paper offers an effective  procedure  to construct algorithmically separating variables admitted by  Hamiltonian integrable systems. As a paradigmatic example, we shall study the case of a system showing an irrational dependence on its coordinates. Precisely, we shall consider a three-parametric deformation of the Holt potential, that has been introduced in \cite{CCR}:
\begin{equation} 
H_1=\frac{1}{2}(p_x^2+p_y^2)+k_1 \frac{4x^2+3y^2}{y^{2/3}}+
k_2\frac {x}{y^{2/3}}+\frac{k_3}{y^{2/3}}. \label{CCR}
\end{equation}
It is integrable in the manifold $M=T^*(E^2 \setminus \{y=0\})$, with a third-order integral
\begin{eqnarray} \label{eqDH2}
 H_2&=&2p_x^3+3p_xp_y^2+
12k_1\bigg(\frac{2x^2-3y^2}{y^{2/3}}+ 6 x y^{1/3}p_y \bigg)+ \\
\nonumber &+&k_2 \bigg( \frac{6x}{y^{2/3}}p_x+9 y^{1/3}p_y\bigg)+  6\frac{k_3}{{y^{2/3}} }p_x \ .
\end{eqnarray}

\noindent When $k_1\rightarrow 0$, the Hamiltonian $H_1$ converts into the (generalized)  Post-Winternitz (gPW) superintegrable potential \cite{PW2011}. A crucial aspect is that $H_1$  in the literature has been regarded as an example of \textit{nonseparable system}. Indeed, it is not of  St\"ackel type; thus is not separable by means of an extended--point transformation. A natural question is whether there is a full canonical transformation (in general difficult to find)   redeeming its separability. Our theory of integrability \textit{\`a la Haantjes}  enables us to solve this problem, since it provides us with a set of DH separation coordinates, according to Theorem \ref{th:SoVgLc}.
\begin{remark} The set of conditions we obtain provides us with an \textit{overdetermined} system of equations.  In general, solving an overdetermined system is computationally affordable; however,  an interesting open question, theoretically relevant, is to find a minimal set of \textit{independent} conditions ensuring separability. 
\end{remark}
\subsection{The $\omega \mathscr{H}$ structure}
Following the procedure of Section  \ref{ss:GP}
we get  a  Haantjes operator   linear in the momenta. It  reads

\begin{equation} \label{eq:LDH}
\boldsymbol{K}^{(DH)}=3
\left[
\begin{array}{cc|cc}
2p_x &  p_y & 0 & 3y  \\
 0& 2p_x & -3y& 0 \\
 \hline
0 & -24 k_1 y^{1/3} & 2p_x &   0\\
24 k_1 y^{1/3} & 0 &  p_y &   2p_x
\end{array}
\right ].
\end{equation}

Due to Lemma \ref{lm:gc}, the Haantjes  operator \eqref{eq:LDH} is a generator of the algebra $ \langle \bs{I}, \bs{K}^{(DH)} \rangle$ since it is semisimple and its minimal polynomial $m_{\boldsymbol{K}}(z)=z^2-b_1(\boldsymbol{x}) z-b_2(\boldsymbol{x})$
 is of degree two. Moreover, its eigenvalues $l_1(\boldsymbol{x})$ and $l_2(\boldsymbol{x})$ fulfil the following recursion relations
\begin{equation} \label{eq:autocv}
(\boldsymbol{K}^{(DH)})^{T} \rd l_1=l_2 \,\rd l_1 \ ,\qquad (\boldsymbol{K}^{(DH)})^{T}\, \rd l_2=l_1 \rd l_2 \ ,
\end{equation}
 that is, they are potential functions of exact eigen 1-forms of $(\boldsymbol{K}^{(DH)})^{T}$. Furthermore, the 
recursion relations \eqref{eq:autocv}  imply the remarkable relation 
\begin{equation} \label{eq:Kowc}
(\boldsymbol{K}^{(DH)})^{T} \rd \left(\frac{1}{2} \,Trace(\boldsymbol{K}^{(DH)})\right )= \rd \sqrt{\det(\boldsymbol{K}^{(DH)}})\ ,
\end{equation}
which  reminds us  the main property of the Kowaleskaya operator appearing in the recent new theory of SoV due to F. Magri \cite{Mp}.  Equation 
\eqref{eq:Kowc} implies that the two functions $\frac{1}{2} \,Trace(\boldsymbol{K}^{(DH)})=b_1(\boldsymbol{x})$ and $\sqrt{\det(\boldsymbol{K}^{(DH)}})=-b_2(\boldsymbol{x})$ are potential functions of exact 1-forms belonging  to a Haantjes chain. Therefore, they are in involution with respect to the Poisson bracket induced by the symplectic form $\omega$, thanks to Lemma \ref{lm:MHchainInv}.  Consequently, by rewriting eq. \eqref{eq:Kowc} in terms of the eigenvalue fields of  $\boldsymbol{K}^{(DH)}$ it is easy to prove that 
\begin{equation} \label{eq:l12PB}
\{ l_1(\boldsymbol{x}),l_2(\boldsymbol{x})\}=0 \ .
\end{equation}



\begin{remark} \label{rem:Hmod2gl}
 Let $\boldsymbol{L}$  be an operator expressed as an affine function of a Haantjes operator $ \boldsymbol{K}$:
 \begin{equation} \label{eq:NK}
\boldsymbol{L}=f \boldsymbol{I}+g\boldsymbol{K} \ ,
 \end{equation}
 where $f$ and $g$ are arbitrary smooth functions, with $g$ nowhere vanishing. Thus, the operator $\boldsymbol{L}$ is  a Haantjes operator as a consequence of identity \eqref{eq:LtorsionLocal}. From Proposition 1.2.4 of \cite{Rob} it follows that the minimal polynomials of the two operators $\boldsymbol{L}$ and $\boldsymbol{K}$ are related by 
 $$
 m_{\boldsymbol{L}}(z)=g^{deg(m_K(z))} m_{\boldsymbol{K}} \big(\frac{1}{g}(z-f)\big) \ ;
  $$
therefore, they have the same number of irreducible factors with the same multiplicity.
 Therefore, the eigenvalues $\lambda_i$ of $\boldsymbol{L}$, the eigenvalues $l_i$ of $\boldsymbol{K}$ and their eigendistributions are related by the  equations
\begin{eqnarray} 
\lambda_i&=&f+g\, l_i 
\label{eq:autovaloriKL} \\
\ker \Bigl(\boldsymbol{L}-\lambda_i\boldsymbol{I}\Bigr)^{\rho_i}&=&\ker \Bigl(\boldsymbol{K}-l_i\boldsymbol{I}\Bigr)^{\rho_i} 
 \label{eq:autovetoriKL}
\end{eqnarray}
for $i=1,\ldots,s$ .
\end{remark}

\begin{remark} \label{rem:QBH}
For the large class represented by the quasi-bi-Hamiltonian systems \cite{BCRR,MT,MTPLA,MTRomp} with two degrees of freedom, we can establish a simple connection with the present theory. In fact, due to the results of \cite{TT2016SIGMA}, we can prove that if the Haantjes operator $\boldsymbol{L}$ in eq. \eqref{eq:NK}  is also a Nijenhuis operator with 
\begin{equation} \label{eq:fgQBH}
f= -\frac{1}{2}\,trace(\boldsymbol{K} )\ , \qquad\qquad g=1   \ ,
\end{equation}
 then the Hamiltonian system under scrutiny admits a quasi-bi-Hamiltonian formulation.
Therefore,  the eigenvalues of the Haantjes generator $\boldsymbol{K}$  are themselves characteristic functions of the Haantjes web since, by plugging the solution \eqref{eq:fgQBH} into eq. \eqref{eq:autovaloriKL}, we obtain
\begin{equation} \label{eq:lQBH}
\lambda_1=-l_2 \ ,\qquad\qquad \lambda_2=-l_1\ .
\end{equation}
In other words, the spectrum of the Haantjes generator $\boldsymbol{K}$ coincides with the spectrum of the Nijenhuis generator $\boldsymbol{L}$. Therefore,  the eigenvalues of the Haantjes operator   \eqref{eq:LDH}  are  characteristic functions of the characteristic web associated with $\boldsymbol{K}^{(DH)}$. This explains, in the case of the Drach-Holt system, the occurrence of  the recurrence relations \eqref{eq:autocv} and the involution properties \eqref{eq:l12PB}.
\end{remark}

\subsection{Separation coordinates}
These functions read
\begin{equation} \label{eq:la}
 \lambda_1= -6(p_x+3\sqrt{2k_1}y^{2/3}), \
 \lambda_2= -6(p_x-3\sqrt{2k_1}y^{2/3}) \ .
\end{equation}
They are in involution with respect to the Poisson bracket induced by $\omega$, due to eq. \eqref{eq:l12PB}. In order to get a system of DH coordinates, we apply the procedure of Section \ref{sec:cOmH}. 

\vspace{2mm}

1) We have that $T_{\boldsymbol{x}}^*M=\mathcal{E}_1^\circ(\boldsymbol{x})\oplus \mathcal{E}_2^\circ(\boldsymbol{x})$, where
 \begin{equation}
 \mathcal{E}_1^\circ= \ker \Bigl((\boldsymbol{K}^{(DH)})^{T} -l_1\boldsymbol{I}\Bigr)=\langle\rd \lambda_1,\gamma_1\rangle \ , \quad
 \mathcal{E}_2^\circ= \ker \Bigl((\boldsymbol{K}^{(DH)})^{T} -l_2\boldsymbol{I}\Bigr)=\langle \rd \lambda_2,\gamma_2\rangle
  \end{equation} with 
\begin{eqnarray*}
\gamma_1&=& -6\,   \sqrt {2 k_1} \, y^{1/3} \rd x 
+\frac {p_y\,y^{1/3}+4\,\sqrt {2 k_1}x}{y^{2/3} }\rd y\,
+2\,{\frac {x}{y^{1/3}}}\, \rd p_x 
+3y^{2/3} \rd p_y \ ,\\
\gamma_2&=& 6\,   \sqrt {2 k_1} \, y^{1/3} \rd x 
+\frac {p_y\,y^{1/3}-4\,\sqrt {2 k_1}x}{y^{2/3} }\rd y\,
+2\,{\frac {x}{y^{1/3}}}\, \rd p_x 
+3y^{2/3} \rd p_y \ .
\end{eqnarray*}
2)
We find two exact $1$-forms $\alpha_1=\rd \mu_1=f_1 \rd \lambda_1+g_1\gamma_1$, $\alpha_2=\rd \mu_2=f_2 \rd \lambda_2+g_2\gamma_2$ such that their potential functions $(\mu_1,\mu_2)$ fulfill the canonical  relations
\begin{equation}
\{ \lambda_1, \mu_1 \}=\langle \rd \lambda_1, P_0 \rd \mu_1 \rangle=1 \ ,\qquad \{\lambda_2 ,\mu_2 \}=\langle \rd \lambda_2, P_0 \rd \mu_2 \rangle=1\ .
\end{equation}
3) Such algebraic equations provide us with the solutions 
$$
g_1= -\frac{1}{72 \sqrt{2 k_1} y^{1/3}}\ , \qquad g_2=  \frac{1}{72 \sqrt{2 k_1} y^{1/3}} \ .
$$
4)
Now, requiring that the 1-forms $ \alpha_1$ and $ \alpha_2$ are closed, we find a  family of solutions for $(f_1,f_2)$. The simplest choices are 
 \begin{equation}
f_1=-\frac{ x} {216 \sqrt{2k_1}y^{2/3} }\ , \qquad
f_2=\frac{ x} {216 \sqrt{2k_1}y^{2/3}} \ .
\end{equation}
(5) Thus, we get the following momenta

\begin{equation}\label{eq:mu}
\mu_1=3 p_y y^{1/3}-6\sqrt{2 k_1} x\ , \qquad
\mu_2=3p_y y^{1/3}+6\sqrt{2 k_1} x \ .
\end{equation}
They are, by construction,  canonically conjugated to the coordinates $(\lambda_1,\lambda_2)$.

\subsection{Separation equations of Jacobi--Sklyanin}
The Jacobi-Sklyanin  approach represents a fundamental piece in the theory of separable systems.  Here we will establish a connection between the Haantjes geometry and the Jacobi-Sklyanin separation equations for the case of the Drach-Holt system.

These  equations allow one to construct a solution $(W,E)$ of the Hamilton-Jacobi equation, according to Eq \eqref{eq:W}. 
The set of coordinates \eqref{eq:la}, \eqref{eq:mu}, being DN coordinates, are separation variables for both the Hamiltonian functions $H_1$ and $H_2$. Besides,  they fulfill the  Jacobi-Sklyanin separation equations
$$
\begin{array}{ll}
 b_1 \mu_1^2+ b_2 \mu_1+b_3\lambda_1^3+\lambda_1 H_1+H_2+b_4&= 0 \ ,\\
 b_1 \mu_1^2+ b_2 \mu_1-b_3\lambda_1^3-\lambda_1 H_1-H_2+b_4&= 0  \ ,
\end{array}
$$
where $b_i$, $i=1,\ldots,7$ are the constants given by
$$
b_1=-10368\sqrt{2k_1^3},\quad b_2=-216 \sqrt{2k_1}\,k_2 ,\quad 
  b_3=\frac{1}{216}, \quad b_4=18 \sqrt{2k_1} k_3 \ .
$$
We arrive therefore at the separated solutions of the Hamilton-Jacobi equation
$$
\begin{array}{ll}
&W_1(\lambda_1;  h_1, h_2)= \frac{1}{2b_1}
\int^{\lambda_1} {\Big(-b_2\pm \sqrt{b_2^2-4b_1P_3(\lambda_1)} \bigg)\rd \lambda_1}, \\
&W_2(\lambda_2; h_1, h_2)=\frac{1}{2b_1}
\int^{\lambda_1} {\Big(-b_2\pm \sqrt{b_2^2-4b_1Q_3(\lambda_1)} \bigg)\rd \lambda_1},
\end{array}
$$
where $h_1,h_2$ are the values of $H_1$, $H_2$ on the Lagrangian tori, and
$$
\begin{array}{ll}
 &P_3(\lambda_1;h_1, h_2 ):=  b_3\lambda_1^3+\lambda_1 h_1+h_2+b_4\ ,\\
&Q_3(\lambda_2;h_1, h_2)=-b_3\lambda_1^3-\lambda_1 h_1-h_2+b_4\ .
\end{array}
$$

\section{Multiseparable systems and Haantjes geometry} \label{sec:multiseparable}

A particularly interesting instance of the previous theory is represented by the case of multiseparable systems. They are Hamiltonian systems that can be separated  in more than one coordinate system in their phase space. 
Fundamental physical examples of multiseparable systems are the $n$-dimensional harmonic oscillator and the Kepler system. 
Another important class is represented by the four Smorodinsky-Winternitz systems, which are the only systems in the Euclidean plane admitting orthogonal separation variables. 

Multiseparable systems are superintegrable ones, provided that the sets of separation functions in the Hamilton-Jacobi equation related to different separation coordinates are functionally independent. 
However, to our knowledge, there is no general theoretical result establishing a  relation between the two notions of superintegrability and separability. For instance, for the classical Post-Winternitz (PW) system \cite{PW2011}, which is maximally superintegrable, no separation coordinates are known in phase space. Besides, due to the presence of integrals of motion of degree higher than two as polynomials in the momenta, the PW system does not admit orthogonal separation coordinates in its configuration space. A superintegrable system can also be simply separable, without being a multiseparable one. This is the case for the anisotropic oscillator discussed in Section \ref{sec:anisotropic}. In this article, we shall focus mainly on the case of superintegrable multiseparable systems, which is perhaps the most interesting one from a geometric and physical point of view.

In order to extend the Haantjes geometry to the case of multiseparable systems, an important, preliminary aspect should be pointed out. 

Let us consider an integrable Hamiltonian system with Hamiltonian function $H$; we denote by $\{(J_k,\phi_k)\}$,   $k =1,\ldots, n$, a set of action-angle variables for the system, with associated frequencies $\nu_{k}(\boldsymbol{J}):=\frac{\partial{H}}{\partial J_{k}}$. In \cite{TT2016prepr}, the Liouville-Haantjes theorem was proved. Under the hypothesis of \textit{nondegeneracy} for $H$, that is
\begin{equation}\label{eq:And}
\det\left(\frac{\partial \nu_k}{\partial J_i} \right)=\det\left(\frac{\partial^2 H}{\partial J_i \partial J_k} \right)\neq 0 \ ,
\end{equation}
this theorem states that  there exists a semisimple $\omega \mathscr{H}$ manifold in any tubular neighbourhood of an Arnold  torus. However, superintegrable systems just violate the condition \eqref{eq:And}; therefore, the LH theorem cannot be applied to them. Nevertheless, this does not imply that $\omega \mathscr{H}$ structures cannot exist for superintegrable systems. Indeed, we can construct such structures by means of a different approach, based on a simple consequence of Theorem \ref{th:SoVgLc}.

\begin{corollary} \label{th:multi}
An integrable Hamiltonian system possesses  as many inequivalent separation coordinate systems as the number of the independent semisimple Abelian $\omega\mathscr{H}$ manifolds of class $n$ that it admits.
\end{corollary} 
\begin{proof}
It suffices to observe that, according to Theorem \ref{th:SoVgLc}, for each semisimple $\omega \mathscr{H}$ structure of class $n$ admitted by the Hamiltonian system there exists a set of Darboux-Haantjes coordinates, which play the role of separating coordinates  for the corresponding Hamilton-Jacobi equation.  
\end{proof}
The previous result represents, jointly with Theorem \ref{th:SoVgLc}, the main theoretical contribution of this work.

\vspace{3mm}

In the following sections, we shall exhibit explicitly the Haantjes structures associated with celebrated examples of both maximally superintegrable and multiseparable Hamiltonian systems.

\section{Construction of Haantjes operators in $T^*Q$: A novel geometric lift} \label{sec:lift}
The problem of constructing Haantjes operators on an $n$-dimensional manifold $M$ is, in general, a hard one, as it entails solving a system of $(n^2(n-1)/2) $ nonlinear PDE of first order in the $n^2$ unknown components of the operators we wish to determine. 

However, when $M \equiv T^{*}Q$, being $Q$ the configuration space of a given mechanical system, one can plan to simplify the problem of constructing a $(T^{*}Q, \omega, \mathscr{H})$ manifold by means of a geometric procedure which allows us to lift a Haantjes operator $\bs{A}: TQ\rightarrow TQ$ to a suitable Haantjes operator $\hat{\bs{A}}:T(T^{*}Q)\to T(T^{*}Q)$. The approach envisaged, in order to be effective, should preserve the vanishing condition of Haantjes tensors. The procedure we propose to this aim is inspired by the one introduced in \cite{IMM} for the construction of a suitable Nijenhuis operator for the Benenti systems; later it has been successfully formulated and applied in \cite{FP} in the context of the $\omega N$ geometry. In those works, from a technical point of view the \emph{complete lift}  from a manifold  to its cotangent bundle, introduced by Yano in \cite{YI1973}, has been adopted. Although the Yano lift preserves the vanishing of  Nijenhuis tensors, unfortunately this property does not hold true for the Haantjes case: the Yano lift of a Haantjes operator need not be another Haantjes operator. 

In order to overcome this drawback, we propose a novel geometric lifting procedure, which generalizes Yano's one.   To this aim, let us consider an operator $\bs{A}: T Q\rightarrow TQ$  and the canonical projection map $\pi: T^*Q \rightarrow Q$. Let us denote by $\bs{\hat{A}}$ a lift of $\bs{A}$ to the cotangent bundle $T^*Q$ which is required to be \emph{projectable} onto $\bs{A}$, that is to say,
$\bs{A}\, \pi_*\, =\pi_*\, \bs{\hat{A}}\, $. Such a condition is fulfilled if and only if $\bs{\hat{A}}$ takes the following block-matrix form in a Darboux chart $(\bold{q},\bold{p})$:
\begin{equation}
\bs{\hat{A}}=\left[\begin{array}{c|c}\bs{A}(\bs{q}) & \bs{0} \\
\hline \bs{C}(\bs{q},\bs{p}) & \bs{D}(\bs{q},\bs{p})
\end{array}\right]\ ,
\end{equation}
where $\bs{C}$ and $\bs{D}$ are $(n\times n)$  matrices depending possibly on all coordinates. In addition, we must impose the compatibility condition with the symplectic form \eqref{eq:gen}, which reduces the form of $\bs{\hat{A}}$ to
\begin{equation} \label{eq:Alift}
\bs{\hat{A}}=\left[\begin{array}{c|c}\bs{A}(\bs{\bs{q}}) & \bs{0} \\
\hline \bs{C(\bs{q},\bs{p})} & \bs{A}^T(\bs{q})
\end{array}\right] \qquad \bs{C}+\bs{C}^T=\boldsymbol{0} \ .
\end{equation}
Assuming that $\bs{A}$ is a Haantjes operator, we wish to determine the form of the matrix $\bs{C}$ in such a way that 
$\bs{\hat{A}} $ is also a Haantjes operator. In the subsequent discussion, we shall prove our result  for configuration spaces of dimension $n=2$ only.
\par
 As clarified in Example 1, when dim $Q=2$, any operator $\bs{A}(\bs{q})$ is a Haantjes operator. Then, its lifted operator takes the form
\begin{equation}
\bs{\hat{A}}=
\left[\begin{array}{cc|cc}
a & b  & 0 & 0 \\
c & d & 0 & 0 \\ \hline 
0 & r & a & c \\ 
-r & 0 & b & d
\end{array}
\right] \ ,
\end{equation}
where all the entries are functions of $(q^{1},q^{2})$ only, except for the smooth function $r=r(q^1,q^2,p_1,p_2)$.  Let us assume that the eigenvalues of $\bs{A}(\bs{q})$ are pointwise distinct; by requiring that $\bs{\hat{A}}$ is a Haantjes operator we find that the most general  solution for $r$  is an affine function of $(p_1,p_2)$. Precisely, $$r=f(q^1,q^2)p_1+g(q^1,q^2)p_2+h(q^1,q^2),$$ where
\begin{eqnarray}
f & = &\frac{\partial a}{\partial q^2}-\frac{\partial b}{\partial q^{1}}-\frac{a-d}{\Delta}\tau^1_{12}-\frac{2b}{\Delta}   \tau^2_{12} \ ,  \\
g & = &\frac{\partial c}{\partial q^2}-\frac{\partial d}{\partial q^1}+ \frac{a-d}{\Delta}\tau^2_{12}  -\frac{2c}{\Delta}   \tau^1_{12} \ ,
\end{eqnarray}
$\Delta:=(a-d)^2+4bc$ is the discriminant of the minimal polynomial of $\bs{A(q)}$, $\tau^1_{12}$ and $ \tau^2_{12}$ are the two independent components of the Nijenhuis torsion of $\bs{A}(\bs{q})$ and $h(q^1,q^2)$ is an arbitrary smooth function. In particular, when $\bs{A}(\bs{q})$ is a Nijenhuis operator and $h(q^1,q^2)=0$, the operator $\bs{\hat{A}} $ coincides with the Yano complete lift of $\bs{A}(\bs{q})$, therefore it is still a Nijenhuis operator. 


The lifting procedure presented here can be extended to the $n$-dimensional case, $n\geq 3$. 
 The details concerning this extension will be discussed elsewhere.

\section{Haantjes structures for multiseparable systems in $E_2$} \label{sec:construction}

The Smorodinsky-Winternitz (SW) systems are a family of superintegrable systems defined in the Euclidean plane $E_2$, which were introduced
first as quantum-mechanical systems in \cite{FMSUW,MSVW, Winternitz} and later studied from a group theoretical point of view in \cite{STW,TTW}. They are all multiseparable in $E_2$ and admit three independent integrals of motion, expressed in terms of second-degree polynomials in the momenta. Also, they are separable in at least two different orthogonal coordinate systems in their configuration space. 

In \cite{TT2012}, the SW systems were analyzed in the context of Nijenhuis geometry. Precisely, it was shown that an $\omega N$ structure can be associated with each of them; this can be achieved by renouncing the standard notion of Lenard chain, and using a generalized version of it. In this section, we will show that a natural and more general framework for studying the geometry of SW systems is offered by the Haantjes geometry. As we will show, one can introduce Haantjes chains and construct two different $\omega \mathscr{H}$ structures for each of the SW systems. In turn,  according to Theorem \ref{th:SoVgLc} and Corollary \ref{th:multi}, these structures guarantee the existence of separating variables for the SW systems.

\subsection{The Smorodinsky-Winternitz system SWI}
The Hamiltonian function is
\begin{equation}
H = H_{1} = \dfrac{1}{2} (p_{x}^{2} + p_{y}^{2}) + \dfrac{1}{2} a (x^{2} + y^{2}) + \dfrac{c_{1}}{x^{2}} + \dfrac{c_{2}}{y^{2}}.
\end{equation}
This system is separable in \textit{cartesian}, \textit{polar} and \textit{elliptic} coordinates, and admits the integrals
\begin{align}
& H_{2} = \dfrac{p_{y}^{2}}{2} + \dfrac{a}{2} y^{2} + \dfrac{c_{2}}{y^{2}}, \\
&H_{3} =  2 \left( c^{2} p_{x}^{2} + (x p_{y} - y p_{x})^{2} + a c^{2} x^{2} + 2 c_{1} \frac{y^{2} + c^{2}}{x^{2}} + 2 c_{2} \left( \frac{x}{y} \right)^{2} \right) \ ,
\end{align}
with $a$, $c$, $c_1$ and $c_2 \in \mathbb{R}$.
A first Haantjes structure is $(\omega, \bs{I}, \bs{K}_{2}^{(SWI)})$, where $\bs{I}$ is the identity operator and 
\begin{equation}
\mathbf{K}_{2}^{(SWI)} = diag(0,1,0,1) \ .
\end{equation}
The algebra $\langle \bs{I}, \bs{K}_{2}^{(SWI)} \rangle$ admits the Cartesian coordinates as DH coordinates: therefore they are obviously separation coordinates for the SWI system.

\vs

We can also obtain a second Haantjes structure $(\omega, \bs{I}, \bs{K}_{3}^{(SWI)})$ with
\begin{equation}
\bs{K}_{3}^{(SWI)} = 4
\left[\begin{array}{cc|cc}
y^{2} + c^{2} & - x y & 0 & 0 \\
- x y & x^{2} & 0 & 0 \\ \hline
0 & - (x p_{y} - y p_{x}) & y^{2} + c^{2} & - x y \\
x p_{y} - y p_{x} & 0 & - x y & x^{2} \\
\end{array}\right].
\end{equation}
The algebra $ \langle\bs{I}, \bs{K}_{3}^{(SWI)} \rangle$ diagonalizes in elliptic coordinates (separation coordinates). If $c=0$, this algebra diagonalizes in polar coordinates. 

\subsection{The Smorodinsky-Winternitz system SWII}
The Hamiltonian function reads:

\begin{equation}
H = H_{1} = \dfrac{1}{2} (p_{x}^{2} + p_{y}^{2}) + a (4 x^{2} + y^{2}) + c_{1} x + \dfrac{c_{2}}{y^{2}} \ ,
\end{equation}
with $a$, $c_1$, $c_2\in \mathbb{R}$. 
The associated integrals are
\begin{align}
& H_{2} = \dfrac{p_{y}^{2}}{2} + a y^{2} + \dfrac{c_{2}}{y^{2}}, \\
& H_{3} = p_{y} (y p_{x} - x p_{y}) + 2 a x y^{2} + \dfrac{c_{1}}{2} y^{2} - 2 c_{2} \dfrac{x}{y^{2}}\ .
\end{align}

We obtain the manifolds $(\omega, \bs{I}, \bs{K}_{2}^{(SWII)})$ and $(\omega, \bs{I}, \bs{K}_{3}^{(SWII)})$, where
\begin{equation}
\bs{K}_{2}^{(SWII)} = diag(0,1,0,1),
\end{equation}
and
\begin{equation}
\bs{K}_{3}^{(SWII)} =
\left[\begin{array}{cc|cc}
0 & y & 0 & 0 \\
y & - 2 x & 0 & 0 \\ \hline
0 & p_{y} & 0 & y \\
- p_{y} & 0 & y & - 2 x \\
\end{array}\right] \ .
\end{equation}
The first structure admits the Cartesian coordinates as DH coordinates,  whereas the second one diagonalizes in parabolic coordinates.

\subsection{The Smorodinsky-Winternitz system SWIII}

The system is defined in polar coordinates $(r, \theta, p_{r}, p_{\theta})$  by the Hamiltonian
\begin{equation}
H = H_{1} = \dfrac{1}{2} \left(p_{r}^{2} + \dfrac{p_{\theta}^{2}}{r^{2}}\right) + \dfrac{\alpha}{r} + \dfrac{1}{r^{2}} \dfrac{\beta + \gamma \cos \theta}{\sin^{2} \theta},
\end{equation}
where $\alpha$, $\beta$, $\gamma\in \mathbb{R}$. The integrals of motion are
\begin{align}
& H_{2} = \dfrac{p_{\theta}^{2}}{2} + \dfrac{\beta + \gamma \cos \theta}{\sin^{2} \theta} \ , \\
& H_{3} = - p_{\theta} \left( \dfrac{p_{\theta} \cos \theta}{r} + p_{r} \sin \theta \right) - \alpha \cos \theta - \dfrac{\gamma + 2 \beta \cos \theta + \gamma \cos^{2} \theta}{r \sin^{2} \theta} \ .
\end{align}

We can construct the Haantjes manifolds $(\omega, \bs{I}, \bs{K}_{2}^{(SWIII)})$ and $(\omega, \bs{I}, \bs{K}_{3}^{(SWIII)})$, with
\begin{equation}
\bs{K}_{2}^{(SWIII)} = diag(0,r^{2},0,r^{2}),
\end{equation}
\begin{equation}
\bs{K}_{3}^{(SWIII)} = -
\left[\begin{array}{cc|cc}
0 & r^{2} \sin \theta & 0 & 0 \\
\sin \theta & 2 r \cos \theta & 0 & 0 \\ \hline
0 & p_{\theta} \cos \theta & 0 & \sin \theta \\
- p_{\theta} \cos \theta & 0 & r^{2} \sin \theta & 2 r \cos \theta \\
\end{array}\right] \ .
\end{equation}
The algebra $ \langle\bs{I}, \bs{K}_{2}^{(SWIII)} \rangle$ ensures the separability of the system in polar coordinates, which are indeed DH coordinates. Notice that if we re-write the operator $\bs{K}_{2}^{(SWIII)}$  in Cartesian coordinates, it coincides (up to an irrelevant multiplicative factor) with the operator $\bs{K}_{3}^{(SWI)}$, for $c=0$. Besides, the algebra $ \langle\bs{I}, \bs{K}_{3}^{(SWIII)} \rangle$ diagonalizes in parabolic coordinates. If we write the expression of $\bs{K}_{3}^{(SWIII)}$ in cartesian coordinates, it converts into the form of the operator $\bs{K}_{3}^{(SWII)}$.

\subsection{The Smorodinsky-Winternitz system SWIV}

The Hamiltonian function of this system is
\begin{equation}
H = H_{1} = \dfrac{1}{2} \dfrac{p_{\xi}^{2} + p_{\eta}^{2}}{\xi^{2} + \eta^{2}} + \dfrac{2 \alpha + \beta \xi + \gamma \eta}{\xi^{2} + \eta^{2}},
\end{equation}
where we have used the parabolic coordinates
\begin{equation}
x = \dfrac{1}{2} (\xi^{2} - \eta^{2}), \hspace{0.5 cm} y = \xi \eta.
\end{equation}

The corresponding integrals read
\begin{align}
& H_{2} = \dfrac{\gamma \xi^{3} + \xi^{2} (p_{\xi} p_{\eta} - \beta \eta) - \xi \eta (p_{\xi}^{2} + p_{\eta}^{2} + 4 \alpha + \gamma \eta) + \eta^{2} (p_{\xi} p_{\eta} + \beta \eta)}{\xi^{2} + \eta^{2}}, \\
& H_{3} = \dfrac{\xi^{2} \left( p_{\eta}^{2} + 2 ( \alpha + \gamma \eta ) - 2 \beta \xi \eta^{2} - \eta^{2} (p_{\xi}^{2} + 2 \alpha) \right)}{\xi^{2} + \eta^{2}}.
\end{align}

In the coordinates $(\xi, \eta, p_{\xi}, p_{\eta})$ we get the Haantjes structures $(\omega, \bs{I}, \bs{K}_{2}^{(SWIV)})$ and $(\omega, \bs{I}, \bs{K}_{3}^{(SWIV)})$, with 
\begin{equation}
\bs{K}_{2}^{(SWIV)} =
\left[\begin{array}{cc|cc}
- 2 \xi \eta & \xi^{2} + \eta^{2} & 0 & 0 \\
\xi^{2} + \eta^{2} & - 2 \xi \eta & 0 & 0 \\ \hline
0 & 0 & - 2 \xi \eta & \xi^{2} + \eta^{2} \\
0 & 0 & \xi^{2} + \eta^{2} & - 2 \xi \eta \\
\end{array}\right],
\end{equation}
\begin{equation}
\bs{K}_{3}^{(SWIV)} = diag(-2 \eta^{2},2 \xi^{2},-2 \eta^{2},2 \xi^{2})\ .
\end{equation}
These two structures are related with two different parabolic  coordinate systems, with distinct axes, which are separating coordinates for the systems SWIV. Notice that the operator $\bs{K}_{3}^{(SWIV)}$ in Cartesian coordinates converts into the operator $-2\bs{K}_{3}^{(SWII)}$.

\begin{remark}
The form of the Haantjes operators presented above can be geometrically interpreted, in all cases, as the application of the generalized lifting procedure described in Section \ref{sec:lift} to a suitable Haantjes operator on the configuration space $E_2$, related to a specific coordinate system in which they diagonalize. Its lifted version to the $4$-dimensional phase space is again a Haantjes operator, at most linearly depending on the momenta. 
\end{remark}

\section{Anisotropic oscillator with Rosochatius terms} \label{sec:anisotropic}
\par

We shall determine now the $\omega\mathscr{H}$ structures of an important physical model: the two-dimensional anisotropic oscillator with Rosochatius terms \cite{Rosochatius}. This system in the general case with $n$ degrees of freedom has been  studied in \cite{Wojc1985,RTW1,RTW2}, where its maximal superintegrability was established. In particular, in \cite{RTW1} the higher-order (missing)  integral was determined by means of a geometric approach based on the Marsden-Weinstein reduction procedure. The anisotropic oscillator on curved spaces has been introduced and studied in \cite{BHKN}.

\vs

The Hamiltonian function of the system reads
\begin{equation}\label{eq:HanoscRos}
H_{AO}= \frac{1}{2} \bigg(p_x^2+p_y^2 +\nu^2(n_1 x^2+ n_2 y^2)+\frac{c_1}{x^2}+\frac{c_2}{y^2}\bigg) \ ,
\end{equation}
where $n_1$, $n_2\in \mathbb{N}\backslash\{0\}$, $c_1$, $c_2$, $\nu\in \mathbb{R}$. The integrals of motion corresponding to the one-dimensional energies are
\begin{eqnarray}
E_1&=& \frac{1}{2} \bigg(p_x^2 +\nu^2 n_1 x^2+\frac{c_1}{x^2}\bigg) \ , \\
E_2&=&\frac{1}{2} \bigg(p_y^2 +\nu^2  n_2 y^2+\frac{c_2}{y^2}\bigg) \ .
\end{eqnarray}
The system \eqref{eq:HanoscRos} is separable in Cartesian coordinates, and admits a Haantjes algebra $\mathscr{H}_1=\langle \boldsymbol{I}, \boldsymbol{K}_{1}^{(AO)}\rangle$, with
\begin{equation}
\bs{K}_{1}^{(AO)} =diag(1,0,1,0).
\end{equation}

However, the system admits a further integral, not related with separating coordinates. Let us introduce the complex quantities
$$\Delta_j=p_j^2+\frac{c_j}{x_j^2}-\nu^2 n_j^2 x_j^2-2i\nu n_j p_j x_j, \qquad \Psi=\Delta_1^{n_2} \bar{\Delta}_2^{n_1} .$$ 
The additional, real integral is given by 
\beq \label{eq:ReQ}
\Re(\Psi)= \Re (\Delta_1^{n_2} \bar{\Delta}_2^{n_1}),
\eeq
where $\Re(z)$ denotes the real part of $z\in \mathbb{C}$.

\vs

We can obtain a second Haantjes algebra in the following way. First, we shall consider the family of operators of the form
\begin{equation}\label{eq:LPW}
\bs{K} =
\left[\begin{array}{cc|cc}
m_{d} & 0 & 0 & 0 \\
m_{21}  & m_{d} & 0 & 0 \\ \hline
0 & m_{32}  & m_{d} & m_{21} \\
-m_{32}  & 0 & 0 & m_{d}  \\
\end{array}\right],
\end{equation}
which are Haantjes operators for any choice of the three arbitrary functions $m_{d}$, $m_{21}$, $m_{32}$ of $(x,y,p_x,p_y)$. This family generalizes (up to a transposition of coordinates with momenta)  the form  of two Haantjes generators admitted by the Post-Winternitz system \cite{TT2019prepr}
\begin{equation}
H^{(PW)} = \dfrac{1}{2} (p_{x}^{2} + p_{y}^{2})  + a \dfrac{x}{y^{2/3}}, \qquad a\in \mathbb{R}.
\end{equation}
Notice that the operators of the form \eqref{eq:LPW} generically are non-semisimple ones.  We shall show that the class \eqref{eq:LPW}  also contains the explicit form of the specific Haantjes operator related to the integral of motion \eqref{eq:ReQ}.

Now, let us impose the chain equation $(\bs{K}_{2}^{(AO)})^{T}~ dH_{AO}=d \Re(\Psi)$, where $\bs{K}_{2}^{(AO)}$ has the form \eqref{eq:LPW}. We find the unique solution 
\begin{eqnarray*}
\big(\bs{K}_{2}^{(AO)})^{T}~dH &=&
(\nu^2(m_{d}n_1^2x+m_{21}n_2^2y)-m_{32}p_y-(\frac{m_{d}c_1}{x^3}+\frac{m_{21}c_2}{y^3})\big)dx \\
& + & (\nu^2m_{d}n_2^2y+m_{32}p_x-\frac{m_{d}c_2}{y^3})dy +
(m_{d}p_x)dp_x \\  &+& (m_{21}p_x+m_{d}p_y) dp_y  \ ,
\end{eqnarray*}
where
\begin{eqnarray}
m_{d}&=&\frac{1}{p_x}\frac{\partial \Re{(\Psi)}}{\partial p_x} \ , \quad p_x \neq 0 \ ; \\
m_{21}&=&\frac{1}{p_x}\bigg(-\frac{p_y}{p_x} \frac{\partial \Re{(\Psi)}}{\partial p_x} +\frac{\partial \Re{(\Psi)}}{\partial p_y} \bigg) \ ; \\
m_{32}&=&\frac{1}{p_x}\bigg(-\nu^2\big(n_2^2 y+\frac{c_2}{y^3}\big)  \frac{1}{p_x}\frac{\partial \Re{(\Psi)}}{\partial p_x}+  
\frac{\partial \Re{(\Psi)}}{\partial y}  
  \bigg)  \ . 
\end{eqnarray}

The explicit form of the coefficients $m_{d}$, $m_{21}$, $m_{32}$ is reported in Appendix \ref{AppA}. Thus, we have proved that the system \eqref{eq:HanoscRos} admits a second, non-semisimple $\omega \mathscr{H}$ structure, with $\mathscr{H}_2= \langle\boldsymbol{I}, \boldsymbol{K}_{2}^{(AO)}\rangle$. It is an open problem to ascertain whether the system \eqref{eq:HanoscRos} admits other separating coordinates in phase space, apart from the Cartesian ones.

\section{$\omega\mathscr{H}$ structures for multiseparable systems in $E_3$} \label{sec:maximal}
We shall study in detail three relevant examples of multiseparable systems in the Euclidean space $E_3$, in the context of Haantjes geometry. One of them is maximally superintegrable, the other two are minimally superintegrable ones \cite{Evans1990}. For a different treatment, in the framework of the Killing-St\"ackel theory, see \cite{Ben93}, \cite{Chanu2001}.
In the following analysis, $\mathcal{L}_x$, $\mathcal{L}_y$, $\mathcal{L}_z$ denote the components of the angular momentum in the Cartesian frame.

\subsection{The Kepler system with a Rosochatius-type term}

We shall consider the Hamiltonian function
\begin{equation} \label{eq:KR}
H = H_{1} = \dfrac{1}{2} (p_{x}^{2} + p_{y}^{2} + p_{z}^{2}) - \dfrac{k}{\sqrt{x^{2} + y^{2} + z^{2}}} + \dfrac{k_{1}}{x^{2}} + \dfrac{k_{2}}{y^{2}} \ ,
\end{equation}
whose integrals of motion are
\begin{align}
& H_{2} = \dfrac{1}{2} |\mathbf{\mathcal{L}}|^{2} + (x^{2} + y^{2} + z^{2}) \left( \dfrac{k_{1}}{x^{2}} + \dfrac{k_{2}}{y^{2}} \right), \\
& H_{3} = \dfrac{1}{2} \mathcal{L}_{z}^{2} + (x^{2} + y^{2}) \left( \dfrac{k_{1}}{x^{2}} + \dfrac{k_{2}}{y^{2}} \right), \\
& H_{4} = \dfrac{1}{2} \mathcal{L}_{y}^{2} + \dfrac{k_{1} z^{2}}{x^{2}}, \\
& H_{5} = \mathcal{L}_{x} p_{y} - p_{x} \mathcal{L}_{y} - 2 z \left( - \dfrac{k}{2 \sqrt{x^{2} + y^{2} + z^{2}}} + \dfrac{k_{1}}{x^{2}} + \dfrac{k_{2}}{y^{2}} \right).
\end{align}
These integrals form three families of functions in involution: $\{H_1, H_2,H_3\}$, $\{H_1, H_2,H_4\}$, $\{H_1, H_3,H_5 \}$.
The equations of the Haantjes chains associated are:
 $\bs{K}_{i}^{T} dH = dH_{i}$ $(i=1,\ldots,5)$, with the Haantjes operators

\begin{equation} \label{eq:K2}
\bs{K}_{2} =
\left[\begin{array}{ccc|ccc}
y^{2} + z^{2} & - x y & - x z & 0 & 0 & 0 \\
- x y & x^{2} + z^{2} & - y z & 0 & 0 & 0 \\
- x z & - y z & x^{2} + y^{2} & 0 & 0 & 0 \\ \hline
0 & - (x p_{y} - y p_{x}) & z p_{x} - x p_{z} & y^{2} + z^{2} & - x y & - x z \\
x p_{y} - y p_{x} & 0 & - (y p_{z} - z p_{y}) & - x y & x^{2} + z^{2} & - y z \\
- (z p_{x} - x p_{z}) & y p_{z} - z p_{y} & 0 & - x z & - y z & x^{2} + y^{2} \\
\end{array}\right] \ ,
\end{equation}

\begin{equation} \label{eq:K3}
\bs{K}_{3} =
\left[\begin{array}{ccc|ccc}
y^{2} & - x y & 0 & 0 & 0 & 0 \\
- x y & x^{2} & 0 & 0 & 0 & 0 \\
0 & 0 & 0 & 0 & 0 & 0 \\ \hline
0 & - (x p_{y} - y p_{x}) & 0 & y^{2} & - x y & 0 \\
x p_{y} - y p_{x} & 0 & 0 & - x y & x^{2} & 0 \\
0 & 0 & 0 & 0 & 0 & 0 \\
\end{array}\right] \ ,
\end{equation}

\begin{equation}
\bs{K}_{4} =
\left[\begin{array}{ccc|ccc}
z^{2} & 0 & - x z & 0 & 0 & 0 \\
0 & 0 & 0 & 0 & 0 & 0 \\
- x z & 0 & x^{2} & 0 & 0 & 0 \\ \hline
0 & 0 & z p_{x} - x p_{z} & z^{2} & 0 & - x z \\
0 & 0 & 0 & 0 & 0 & 0 \\
- (z p_{x} - x p_{z} ) & 0 & 0 & - x z & 0 & x^{2} \\
\end{array}\right] 
\end{equation}
\noi and
\begin{equation} \label{eq:K7}
\bs{K}_{5} =
\left[\begin{array}{ccc|ccc}
- 2 z & 0 & x & 0 & 0 & 0 \\
0 & - 2 z & y & 0 & 0 & 0 \\
x & y & 0 & 0 & 0 & 0 \\ \hline
0 & 0 & - p_{x} & - 2 z & 0 & x \\
0 & 0 & - p_{y} & 0 & - 2 z & y \\
p_{x} & p_{y} & 0 & x & y & 0 \\
\end{array}\right] \ .
\end{equation}

We shall also take into account the integral 
\[
H_6:= \dfrac{1}{2} \mathcal{L}_{x}^{2} + \dfrac{k_{2} z^{2}}{y^{2}} 
\ , 
\]
which is not functionally independent with respect to the other ones; however, it will play a relevant role in the construction of separating coordinates, as we shall see below.
Imposing the chain equation  $\bs{K}_{6}^{T} dH = dH_{6}$, we get the further Haantjes operator
\begin{equation} \label{eq:K6}
\bs{K}_{6} =
\left[\begin{array}{ccc|ccc}
0 & 0 & 0 & 0 & 0 & 0 \\
0 & z^2 & -yz & 0 & 0 & 0 \\
0 & -yz & y^{2} & 0 & 0 & 0 \\ \hline
0 & 0 & 0  & 0 & 0 & 0 \\
0 & 0 & z p_{y}- y p_{z} & 0 & z^2 & -yz \\
0 & y p_{z}-z p_{y} & 0 & 0 & -yz & y^{2} \\
\end{array}\right] \ .
\end{equation}
Let us discuss now the separating coordinates admitted by the system  \eqref{eq:KR}. First, we observe that there exist three Abelian semisimple Haantjes algebras related with the operators \eqref{eq:K2}-\eqref{eq:K7}: $\mathscr{H}_1= \langle\bs{I}, \bs{K}_{2}, \bs{K}_{3}\rangle$,  $\mathscr{H}_2=\langle\bs{I},\bs{K}_{2},\bs{K}_{4}\rangle$ and $\mathscr{H}_{3}:= \langle\bs{I}, \bs{K}_2, \bs{K}_6\rangle$.
The algebra $\mathscr{H}_1= \langle\bs{I}, \bs{K}_{2}, \bs{K}_{3}\rangle$, as was discussed in Ref. \cite{TT2021}, diagonalizes in the \ti{spherical polar} coordinates with the $z$-axis as the polar axis; therefore, according to the previous discussion, they are DH coordinates and separating coordinates for the Kepler system \eqref{eq:KR}. The algebra $\mathscr{H}_2$ diagonalizes in spherical polar   coordinates with the $y$-axis as the polar axis; besides, the algebra $\mathscr{H}_3$ ensures the separability of the system in spherical polar coordinates with the $x$-axis as the polar axis.

\vs

We also have the algebra $\mathscr{H}_4:= \langle\bs{I}, \bs{K}_{3}, \bs{K}_{5}\rangle$, which diagonalizes in the  \ti{rotational parabolic} coordinates.

\vs

Finally, the system \eqref{eq:KR} separates in \textit{spherical conical} coordinates. Precisely, the algebra $\mathscr{H}_5:=\langle\bs{I}, \bs{K}_2, \bs{K}_7:= a \bs{K}_6+ b \bs{K}_4 + c \bs{K}_3 \rangle$ diagonalizes in the spherical conical coordinates $u_1\in \mathbb{R}$, $u_2\in \mathbb{R}$, $u_3>0$, where $a\leq u_1 \leq b \leq u_2 \leq c$ and $a,b,c\in \mathbb{R}$.

\subsubsection{Generators of the cyclic Haantjes algebras.} In order to find a generator of a cyclic semisimple Abelian Haantjes algebra of rank $n$, it is sufficient to choose inside the algebra a Haantjes operator whose minimal polynomial is of degree $n$, as explained in Section \ref{sec:cOmH}.
As a consequence of this observation, we obtain that a simple (not unique) choice for the generator for the algebra $\mathscr{H}_1=\langle\bs{I}, \bs{K}_{2}, \bs{K}_{3}\rangle$ reads
\beq \label{eq:CGI}
\bs{L}_{1}= \bs{K}_{2} + \bs{K}_{3} \ .
\eeq
We have
\beq \label{eq:K4C}
\bs{K}_{2}= \al_1^{(1)} \bs{L}_{1}+ \al_2^{(1)} \bs{L}_{1}^{2} \ , \qquad 
\al_1^{(1)}= \f{3 x^2+3 y^2+2 z^2}{2x^2+2y^2+z^2} \ , \qquad \al_2^{(1)}= - \f{1}{2x^2+2y^2+z^2} \ ,
\eeq
and
\beq \label{eq:K5C}
\bs{K}_{3} = \be_1^{(1)} \bs{L}_{1}+ \be_2^{(1)} \bs{L}_{1}^{2} \ , \qquad 
\be_1^{(1)}= - \f{x^2+y^2+z^2}{2x^2+2y^2+z^2}, \qquad \be_2^{(1)}=  \f{1}{2x^2+2y^2+z^2} \ .
\eeq

\vs

A generator of the algebra $\mathscr{H}_2 = \langle\bs{I}, \bs{K}_{2}, \bs{K}_{4}\rangle$ is given by
\begin{equation}
\bs{L}_2= \bs{K}_{2} + \bs{K}_{4} \ .
\end{equation}
We have
\begin{equation}
\bs{K}_{2}= \alpha_1^{(2)} \bs{L}_2+ \alpha_2^{(2)} \bs{L}_{2}^{2} \ ,
\qquad 
\alpha_1^{(2)}= \frac{3 x^2+2 y^2+3 z^2}{2 x^2+y^2+2 z^2} \ , \qquad \alpha_2^{(2)}= -\frac{1}{2 x^2+y^2+2 z^2} \ ,
\end{equation}
and
\begin{equation}
\bs{K}_{4} = \beta_1^{(2)} \bs{L}_2+ \beta_2^{(2)} \bs{L}_{2}^{2} \ , \qquad
\beta_1^{(2)}= -\frac{x^2+y^2+z^2}{2 x^2+y^2+2 z^2} \ , \qquad \beta_2^{(2)}=  \frac{1}{2 x^2+y^2+2 z^2} \ .
\end{equation}
A generator of the algebra $\mathscr{H}_3=\langle\bs{I}, \bs{K}_2, \bs{K}_6 \rangle$ is $\bs{L}_3= \bs{K}_2 + \bs{K}_6$, with
\begin{equation}
\bs{K}_{2}= \alpha_1^{(3)} \bs{L}_3+ \alpha_2^{(3)} \bs{L}_{3}^{2} \ ,
\qquad 
\alpha_1^{(3)}= \frac{2 x^2+3 y^2+3 z^2}{ x^2+2y^2+2 z^2} \ , \qquad \alpha_2^{(3)}= -\frac{1}{ x^2+2y^2+2 z^2} \ ,
\end{equation}
and
\begin{equation}
\bs{K}_{6} = \beta_1^{(3)} \bs{L}_3+ \beta_2^{(3)} \bs{L}_{3}^{2} \ , \qquad
\beta_1^{(3)}= -\frac{x^2+y^2+z^2}{x^2+ 2y^2+2 z^2} \ , \qquad \beta_2^{(3)}=  \frac{1}{ x^2+2y^2+2 z^2} \ .
\end{equation}

A generator of the algebra $\mathscr{H}_4 = \langle\bs{I}, \bs{K}_{3}, \bs{K}_{5}\rangle$ is just  $\bs{L}_4= \bs{K}_{5}$, since the minimal polynomial of $\bs{K}_5$ is of degree $3$. We have

\[
\bs{K}_3= \al_0^{(4)} \bs{I} + \al_1^{(4)} \bs{L}_4 + \al_2^{(4)} \bs{L}_{4}^2 \ ,
\]
with
\[
\al_0^{(4)} = x^{2}+y^2, \qquad \al_1^{(4)}= -2 z, \qquad \al_3^{(4)} =-1 \ . 
\]
A generator of the algebra $\mathscr{H}_5 =\langle\bs{I}, \bs{K}_2, \bs{K}_7\rangle$ is simply $\bs{L}_5=\bs{K}_7$, with
\[
\bs{K}_2= \al_{1}^{(5)} \bs{L}_5+ \al_{2}^{(5)} \bs{L}_5^{2} \ ,
\]
where
\[
\al_1^{(5)}=  \frac{(b+c)x^2+(a+c)y^2+(a+b)z^2}{bcx^2+ac y^2+abz^2} \qquad \al_2^{(5)} = -\frac{1}{c(bx^2+ay^2)+ab z^2} \ .
\]

\subsection{A class of generalized Kepler systems}

We consider a family of deformations of the Kepler system which depends on an arbitrary function $F \left( \frac{y}{x}\right)$. The Hamiltonian function of this family has the form
\begin{equation}\label{eq:gK}
H = H_{1} = \dfrac{1}{2} (p_{x}^{2} + p_{y}^{2} + p_{z}^{2}) - \dfrac{k}{\sqrt{x^{2} + y^{2} + z^{2}}} + \dfrac{k_{1} z}{\sqrt{x^{2} + y^{2} + z^{2}} (x^{2} + y^{2})} + \dfrac{F \left( \frac{y}{x} \right)}{x^{2} + y^{2}}.
\end{equation}

The associated integrals of motion are
\begin{align}
& H_{2} = \dfrac{1}{2} |\mathbf{\mathcal{L}}|^{2} + \dfrac{k_{1} z \sqrt{x^{2} + y^{2} + z^{2}} + (x^{2} + y^{2} + z^{2}) \  F\left( \frac{y}{x} \right)}{x^{2} + y^{2}}, \\
& H_{3} = \dfrac{1}{2} \mathcal{L}_{z}^{2} + F \left( \frac{y}{x} \right), \\
& H_{4} = \mathcal{L}_{x} p_{y} - p_{x} \mathcal{L}_{y} + \dfrac{k z}{\sqrt{x^{2} + y^{2} + z^{2}}} - \dfrac{k_{1} (x^{2} + y^{2} + 2 z^{2})}{\sqrt{x^{2} + y^{2} + z^{2}} (x^{2} + y^{2})} - 2 z \dfrac{F \left( \frac{y}{x} \right)}{x^{2} + y^{2}}.
\end{align}

These integrals form two families of functions in involution: $\{H_1, H_2,H_3\}$ and $\{H_1, H_3,H_4\}$.
The equations of the Haantjes chains are: $\bs{K}_{2}^{T} dH = dH_{2}$, $\bs{K}_{3}^{T} dH = dH_{3}$, $\bs{K}_{5}^{T} dH = dH_{4}$. The Haantjes operators associated are $\bs{K}_{2}$, eq. \eqref{eq:K2},  $\bs{K}_{3}$, eq. \eqref{eq:K3} and  $\bs{K}_{5}$, eq. \eqref{eq:K7}. We obtain again the two Abelian, cyclic algebras $\mathscr{H}_{1} = \langle\bs{I}, \bs{K}_{2}, \bs{K}_{3}\rangle$ and $\mathscr{H}_{4} = \langle \bs{I}, \bs{K}_{3}, \bs{K}_{5}\rangle$, whose generators have been determined above. They provide us with the two coordinate systems admitted by the class of Hamiltonian systems \eqref{eq:gK}, namely the spherical polar and the rotational parabolic ones. 
\subsection{A class of anisotropic oscillators}

An interesting family of deformations of the anisotropic oscillator is given by the Hamiltonian function
\begin{equation} \label{eq:defAO}
H = H_{1} = \dfrac{1}{2} (p_{x}^{2} + p_{y}^{2} + p_{z}^{2}) + k ( x^{2} + y^{2}) + 4 k z^{2} + \dfrac{F\left(\frac{y}{x}\right)}{x^{2} + y^{2}} \ ,
\end{equation}
where again $F$ is an arbitrary function of its argument. It admits the following integrals of motion:
\begin{align}
& H_{2} = \dfrac{1}{2} p_{z}^{2} + 4 k z^{2}, \\
& H_{3} = \dfrac{1}{2} \mathcal{L}_{z}^{2} + F\left( \frac{y}{x} \right), \\
& H_{4} = \mathcal{L}_{x} p_{y} - p_{x} \mathcal{L}_{y} + 2 k z (x^{2}+y^{2}) - 2 z \dfrac{F\left(\frac{y}{x}\right)}{x^{2} + y^{2}}.
\end{align}

Therefore, this class of systems is minimally superintegrable. These integrals form two families of functions in involution: $\{H_1, H_2,H_3\}$ and $\{H_1, H_3,H_4\}$. The equations for the Haantjes chains  are:
$\bs{K}_{1}^{T} dH = dH_{2}$, $\bs{K}_{3}^{T} dH = dH_{3}$, $\bs{K}_{5}^{T} dH = dH_{4}$. The corresponding Haantjes operators are: 

\begin{equation} \label{eq:K1}
\bs{K}_{1} =diag(0,0,1,0,0,1),
\end{equation}
$\bs{K}_{3}$, eq. \eqref{eq:K3} and $\bs{K}_{5}$, eq. \eqref{eq:K7}.
 The family of systems \eqref{eq:defAO} admits the two Abelian algebras 
$\mathscr{H}_{6} := \langle\bs{I}, \bs{K}_{1}, \bs{K}_{3}\rangle$ and $\mathscr{H}_{4} = \langle\bs{I}, \bs{K}_{3}, \bs{K}_{5}\rangle$. It separates in polar cylindrical and rotational parabolic coordinates.

\vs

The algebra $\mathscr{H}_{6}=\langle\bs{I}, \bs{K}_{1}, \bs{K}_{3}\rangle$ admits the generator
\beq
\bs{L}_6= 2 \bs{K}_{1} + \frac{1}{x^{2}+y^{2}} \bs{K}_{3} \ .
\eeq
We have
\beq
\bs{K}_{1}= \al_1^{(6)} \bs{L}_6+ \al_2^{(6)} \bs{L}_{6}^{2} , \qquad
\al_1^{(6)}= - \frac{1}{2}, \qquad \al_2^{(6)}= \frac{1}{2} \ ,
\eeq
and 
\beq
\bs{K}_{3} = \be_1^{(6)} \bs{L}_{6}+ \be_2^{(6)} \bs{L}_{6}^{2}, \qquad
\be_1^{(6)}= 2 ( x^2+y^2 ), \qquad \be_2^{(6)}= - (x^2+y^2 ) \ .
\eeq



\section{Future perspectives} \label{sec:future}

The problem of finding separating variables for integrable Hamiltonian systems is certainly among the most relevant ones of Classical Mechanics. The theory of $\omega \mathscr{H}$ manifolds offers a twofold contribution to this fundamental problem, being of both conceptual and  applicative nature.

From a conceptual point of view, we have shown that the existence of semisimple $\omega\mathscr{H}$ structures ensures that of separating variables. Precisely, as stated in Theorem  \ref{th:SoVgLc}, if an integrable system admits a semisimple Abelian $\omega \mathscr{H}$ structure, then one can construct  a privileged set of coordinates, the Darboux-Haantjes (DH) coordinates, which are   separation coordinates for the Hamilton-Jacobi equation associated with the system. Besides, in these coordinates the symplectic form takes a Darboux form, and the operators of the Haantjes algebra diagonalize simultaneously.


From an applicative point of view, the $\omega \mathscr{H}$ structures represent a very flexible tool, which can be used either  to construct in a consistent way separating coordinates,  or to enhance the applicability of Nijenhuis geometry. In this regard, we have clarified the relationship between $\omega \mathscr{H}$ and $\omega N$ structures. The approach \textit{\`a la Nijenhuis}  represents an important theoretical framework, allowing for the construction of separating variables via the eigenvalues of a suitable semisimple Nijenhuis operator associated with a given integrable system \cite{FP}. The Haantjes geometry can complement and integrate, for practical purposes, the Nijenhuis approach. Indeed, a semisimple Haantjes algebra always admits a Haantjes generator; in addition, such a generator can be chosen to be a Nijenhuis one. Therefore, a possible strategy for finding separating variables is the following:  given a Hamiltonian system, first we can construct in a natural way a $\omega \mathscr{H}$ manifold by solving the equations of its Haantjes chains (without the need for generalized ones, as is often necessary in Nijenhuis geometry \cite{TT2012}). Then, assuming that the algebra $\mathscr{H}$ is semisimple and Abelian, we can always select in $\mathscr{H}$ a Nijenhuis generator, providing us with separating variables (the so called Darboux-Nijenhuis (DN) coordinates). In the multiseparable case, this procedure can be repeated for all the structures allowed by the considered system. 

We mention that finding a Nijenhuis operator without the help of the $\omega \mathscr{H}$ structure can be computationally cumbersome. Usually, one tries to construct Nijenhuis operators as reductions of bi-Hamiltonian structures associated with systems defined in a higher-dimensional space. Another possibility is to generate $\omega N$ structures as Yano liftings of Nijenhuis operators defined in the configuration space. However, these  procedures are not generally applicable: only in some cases one can determine a system (easily tractable within the Nijenhuis approach), that by reduction gives the original system under study (and consequently, its $\omega N$ structure). Also, there are systems, as the Drach-Holt and the Post-Winternitz ones, whose Nijenhuis structure is not obtainable as a Yano complete lift.  

Nevertheless,  as shown in this article (and in many examples of \cite{TT2016prepr},\cite{TT2016SIGMA}), for determining an $\omega \mathscr{H}$  structure    a direct, computationally affordable procedure is available, since the differential equations for the Haantjes chains are in general quite manageable. 

\vs

In short, given an integrable Hamiltonian system, separating coordinates can be found in two different ways: as DH coordinates associated with its $\omega \mathscr{H}$ semisimple structures or, equivalently, via the eigenvalues of the Nijenhuis operator generating the same $\omega \mathscr{H}$ structures (DN coordinates).

Interestingly enough, $\omega \mathscr{H}$ manifolds appear to be a quite ubiquitous geometric structure in the class of superintegrable systems: indeed, they exist also in cases where orthogonal separating variables are not allowed in $T^{*}Q$ and are realized by \textit{non-semisimple} Haantjes operators. This is the case for the anisotropic oscillator described above, and for the Post-Winternitz system \cite{PW2011}, which admits a \textit{non-Abelian} Haantjes algebra, possessing two Abelian subalgebras of non-semisimple Haantjes operators.

We mention that a generalization of the notion of $\omega \mathscr{H}$ manifolds is that of $P\mathscr{H}$ manifolds, introduced in \cite{T2017}, where  $P$ is a  Poisson bivector compatible with the Haantjes algebra $\mathscr{H}$. When $P$ is invertible, a $P \mathscr{H}$ structure reduces to a $\omega \mathscr{H}$ one.

\vs

A fundamental open problem is to study the geometry of non-semisimple $\omega \mathscr{H}$ structures (being in this case more general than $\omega N$ ones), and to ascertain their relevance from the point of view of the general problem of SoV. 

Another interesting open problem is the geometric interpretation of the new, infinite class of generalized torsions introduced in \cite{TT2019prepr} from the perspective of classical Hamiltonian systems.  A crucial result, proved in \cite{TT2019prepr}, states that if the generalized torsion of an operator field vanishes, then its eigen-distributions are mutually integrable. The ultimate implications of this property in the context of the theory of Hamiltonian integrable systems are presently under investigation.

\section*{Acknowledgement}
\noi The research of D. R. N. has been supported by the  Severo Ochoa Programme for Centres of Excellence in R\&D (CEX2019-000904-S), Ministerio de Ciencia, Innovaci\'on y Universidades, Spain.
\noi The research of P. T. has been supported by the research project PGC2018-094898-B-I00, Ministerio de Ciencia, Innovaci\'on y Universidades, Spain, and by the Severo Ochoa Programme for Centres of Excellence in R\&D (CEX2019-000904-S), Ministerio de Ciencia, Innovaci\'on y Universidades, Spain.
\par
The research of G. T. has been supported by the research project FRA2020-2021, Universit\'a degli Studi di Trieste, Italy.

\noi  P. T. is member of Gruppo Nazionale di Fisica Matematica (GNFM) of INDAM.

\appendix

\section{On the Haantjes structure of the system \eqref{eq:HanoscRos}}\label{AppA}

As an illustrative example of the form of the integral of motion $\Re{(\Psi)}$, eq. \eqref{eq:ReQ}, admitted by the system \eqref{eq:HanoscRos}, we shall consider  the case $n_1=1,n_2=3$. The integral  \eqref{eq:ReQ} reads explicitly
\begin{align}
&\Re{(\Psi)}= 9\,{\nu}^{8}{x}^{6}{y}^{2}+27\,{\frac {{c_{1}}^{2}{\nu}^{4}{y}^{2}}{{x
}^{2}}}-27\,c_{1}\,{x}^{2}{\nu}^{6}{y}^{2}+{\frac {{c_{1}}^{3}c_{2}}{{
x}^{6}{y}^{2}}} \\ & -9\,{\frac {{c_{1}}^{3}{\nu}^{2}{y}^{2}}{{x}^{6}}} -{
\frac {{\nu}^{6}{x}^{6}c_{2}}{{y}^{2}}}-3\,{\frac {{c_{1}}^{2}{\nu}^{2
}c_{2}}{{x}^{2}{y}^{2}}}+3\,{\frac {c_{1}\,{x}^{2}{\nu}^{4}c_{2}}{{y}^
{2}}} \nonumber \\
&+ \left( 162\,c_{1}\,{\nu}^{4}{y}^{2}-135\,{\nu}^{6}{x}^{4}{y}^{2
}+3\,{\frac {{c_{1}}^{2}c_{2}}{{x}^{4}{y}^{2}}}-27\,{\frac {{c_{1}}^{2
}{\nu}^{2}{y}^{2}}{{x}^{4}}}-18\,{\frac {c_{1}\,{\nu}^{2}c_{2}}{{y}^{2
}}}+15\,{\frac {{\nu}^{4}{x}^{4}c_{2}}{{y}^{2}}} \right) {{\it p_x}}^{2
} \nonumber \\
& + \left( {\frac {{c_{1}}^{3}}{{x}^{6}}}-{\nu}^{6}{x}^{6}-3\,{\frac {{
c_{1}}^{2}{\nu}^{2}}{{x}^{2}}}+3\,c_{1}\,{x}^{2}{\nu}^{4} \right) {{
\it p_y}}^{2}+ \left( 3\,{\frac {{c_{1}}^{2}}{{x}^{4}}}-18\,c_{1}\,{\nu
}^{2}+15\,{\nu}^{4}{x}^{4} \right) {{\it p_x}}^{2}{{\it p_y}}^{2} \nonumber \\ & +{{\it 
p_x}}^{6}{{\it p_y}}^{2}+ \left( 36\,{\frac {{c_{1}}^{2}{\nu}^{2}y}{{x}^
{3}}}-72\,c_{1}\,x{\nu}^{4}y+36\,{\nu}^{6}{x}^{5}y \right) {\it p_y}\,{
\it p_x} \nonumber \\ & + \left( 135\,{\nu}^{4}{x}^{2}{y}^{2}+3\,{\frac {c_{1}\,c_{2}}{
{x}^{2}{y}^{2}}}-27\,{\frac {c_{1}\,{\nu}^{2}{y}^{2}}{{x}^{2}}}-15\,{
\frac {{\nu}^{2}{x}^{2}c_{2}}{{y}^{2}}} \right) {{\it p_x}}^{4} \nonumber \\
&
+ \left( {\frac {c_{2}}{{y}^{2}}}-9\,{\nu}^{2}{y}^{2} \right) {{\it p_x}
}^{6}+ \left( 72\,{\frac {c_{1}\,{\nu}^{2}y}{x}}-120\,{\nu}^{4}{x}^{3}
y \right) {\it p_y}\,{{\it p_x}}^{3}+ \left( 3\,{\frac {c_{1}}{{x}^{2}}}
-15\,{\nu}^{2}{x}^{2} \right) {{\it p_x}}^{4}{{\it p_y}}^{2} \nonumber \\ & +36\,{\nu}^{
2}{{\it p_x}}^{5}x{\it p_y}\,y \nonumber
\end{align}
and the elements of the Haantjes operator are
\begin{align}
m_{d}=&-270\,{\nu}^{6}{x}^{4}{y}^{2}+324\,c_{1}\,{\nu}^{4}{y}^{2}+30\,{\frac 
{{\nu}^{4}{x}^{4}c_{2}}{{y}^{2}}}-54\,{\frac {{c_{1}}^{2}{\nu}^{2}{y}^
{2}}{{x}^{4}}}-36\,{\frac {c_{1}\,{\nu}^{2}c_{2}}{{y}^{2}}} \\ & +6\,{\frac 
{{c_{1}}^{2}c_{2}}{{x}^{4}{y}^{2}}}+ \left( 540\,{\nu}^{4}{x}^{2}{y}^{
2}-108\,{\frac {c_{1}\,{\nu}^{2}{y}^{2}}{{x}^{2}}}-60\,{\frac {{\nu}^{
2}{x}^{2}c_{2}}{{y}^{2}}}+12\,{\frac {c_{1}\,c_{2}}{{x}^{2}{y}^{2}}}
 \right) {{\it p_x}}^{2} 
\nonumber \\ & + \left( 30\,{\nu}^{4}{x}^{4}-36\,c_{1}\,{\nu}^
{2}+6\,{\frac {{c_{1}}^{2}}{{x}^{4}}} \right) {{\it p_y}}^{2}+ \left( -
60\,{\nu}^{2}{x}^{2}+12\,{\frac {c_{1}}{{x}^{2}}} \right) {{\it p_x}}^{
2}{{\it p_y}}^{2} \nonumber \\
& +{\frac {{\it p_y}}{{\it p_x}} \left( 36\,{\frac {{\nu}^
{2}{c_{1}}^{2}y}{{x}^{3}}}-72\,c_{1}\,x{\nu}^{4}y+36\,{\nu}^{6}{x}^{5}
y \right) }+180\,{{\it p_x}}^{3}xy{\nu}^{2}{\it p_y} \nonumber \\
& + \left( -360\,{\nu}
^{4}{x}^{3}y+216\,{\frac {c_{1}\,{\nu}^{2}y}{x}} \right) {\it p_x}\,{
\it p_y}+ \left( -54\,{\nu}^{2}{y}^{2}+6\,{\frac {c_{2}}{{y}^{2}}}
 \right) {{\it p_x}}^{4}+6\,{{\it p_x}}^{4}{{\it p_y}}^{2} , \nonumber
\end{align}
\begin{align}
m_{21}= &
36\,{\frac {{\nu}^{2}{c_{1}}^{2}y}{{x}^{3}}}-72\,c_{1}\,x{\nu}^{4}y+36
\,{\nu}^{6}{x}^{5}y+ \left( 72\,{\frac {c_{1}\,{\nu}^{2}y}{x}}-120\,{
\nu}^{4}{x}^{3}y \right) {{\it p_x}}^{2} \\ 
 & + \left( 360\,{\nu}^{4}{x}^{3}y
-216\,{\frac {c_{1}\,{\nu}^{2}y}{x}} \right) {{\it p_y}}^{2} \nonumber \\
&+
{\frac {{\it p_y}}{{\it p_x}}
 \left( 36\,{\frac {c_{1}\,{\nu}^{2}c_{2}}{{y}^{2}}}-2\,{\nu}^{6}{x}^{
6}+270\,{\nu}^{6}{x}^{4}{y}^{2}+2\,{\frac {{c_{1}}^{3}}{{x}^{6}}}+6\,c
_{1}\,{x}^{2}{\nu}^{4} \right.} \nonumber \\ & { \left. \qquad \qquad \qquad -324\,c_{1}\,{\nu}^{4}{y}^{2}-30\,{\frac {{\nu}^
{4}{x}^{4}c_{2}}{{y}^{2}}}-6\,{\frac {{c_{1}}^{2}{\nu}^{2}}{{x}^{2}}} +
54\,{\frac {{c_{1}}^{2}{\nu}^{2}{y}^{2}}{{x}^{4}}}-6\,{\frac {{c_{1}}^
{2}c_{2}}{{x}^{4}{y}^{2}}} \right) } \nonumber \\
& + \left( -30\,{\nu}^{2}{x}^{2}+54
\,{\nu}^{2}{y}^{2}+6\,{\frac {c_{1}}{{x}^{2}}}-6\,{\frac {c_{2}}{{y}^{
2}}} \right) {\it p_y}\,{{\it p_x}}^{3}+
 \left( 60\,{\nu}^{2}{x}^{2}-12\,{\frac {c_{1}}{{x}^{2}}} \right) {{
\it p_y}}^{3}{\it p_x} \nonumber \\
& -6\,{{\it p_y}}^{3}{{\it p_x}}^{3} +2\,{{\it p_x}}^{5}
{\it p_y}-180\,{{\it p_x}}^{2}yx{\nu}^{2}{{\it p_y}}^{2} \nonumber
\\&
+
 \left( 30\,{\nu}^{4}{x}^{4}-540
\,{\nu}^{4}{x}^{2}{y}^{2}+6\,{\frac {{c_{1}}^{2}}{{x}^{4}}}-36\,c_{1}
\,{\nu}^{2} \right. \nonumber \\ &  \left. \qquad \qquad \qquad +108\,{\frac {c_{1}\,{\nu}^{2}{y}^{2}}{{x}^{2}}} +60\,{
\frac {{\nu}^{2}{x}^{2}c_{2}}{{y}^{2}}}-12\,{\frac {c_{1}\,c_{2}}{{x}^
{2}{y}^{2}}} \right) {\it p_x}\,{\it p_y} \nonumber \\
&+{\frac {{{\it p_y}}^{3}}{{\it 
p_x}} \left( -30\,{\nu}^{4}{x}^{4}-6\,{\frac {{c_{1}}^{2}}{{x}^{4}}}+36
\,c_{1}\,{\nu}^{2} \right) }+36\,{{\it p_x}}^{4}yx{\nu}^{2} \nonumber \\
& +{\frac {{{
\it p_y}}^{2}}{{{\it p_x}}^{2}} \left( -36\,{\nu}^{6}{x}^{5}y+72\,c_{1}
\,x{\nu}^{4}y-36\,{\frac {{\nu}^{2}{c_{1}}^{2}y}{{x}^{3}}} \right) }, \nonumber
\end{align}
\begin{align}
m_{32}= &
 \left( 324\,y{\nu}^{4}c_{1}+972\,{\frac {{y}^{3}{\nu}^{4}c_{1}}{{x}^{
2}}}-30\,{\frac {{\nu}^{4}{x}^{4}c_{2}}{{y}^{3}}} \right. \\ &  \left. \qquad \qquad \qquad +1080\,{\frac {{\nu}^
{4}{x}^{2}c_{2}}{y}} -54\,{\frac {y{\nu}^{2}{c_{1}}^{2}}{{x}^{4}}}-60\,
{\frac {{\nu}^{2}{x}^{2}{c_{2}}^{2}}{{y}^{5}}}-6\,{\frac {{c_{1}}^{2}c_{2}}{{x}^{4}{y}^{3}}}
\right ) {\it p_x} \nonumber
\\
&+\left (12\,{\frac {c_{1}\,{c_{2}}^{2}}{{x}^{2}{y}^{5}}
}+36\,{\frac {c_{1}\,{\nu}^{2}c_{2}}{{y}^{3}}}
-216\,{\frac {c_{1}\,{
\nu}^{2}c_{2}}{{x}^{2}y}}-270\,{x}^{4}y{\nu}^{6}-
4860\,{x}^{2}{y}^{3}{\nu}^{6} \right) {\it p_x} \nonumber
\\
&+ \left( 
216\,{\frac {c_{1}\,{\nu}^{2}c_{2}}{
x{y}^{2}}}+36\,{\nu}^{6}{x}^{5}+3240\,{x}^{3}{y}^{2}{\nu}^{6} \right. \nonumber \\ &  \left. \qquad \qquad \qquad +36\,{
\frac {{\nu}^{2}{c_{1}}^{2}}{{x}^{3}}}-72\,c_{1}\,x{\nu}^{4}-1944\,{
\frac {c_{1}\,{\nu}^{4}{y}^{2}}{x}}-360\,{\frac {{\nu}^{4}{x}^{3}c_{2}
}{{y}^{2}}} \right ) {\it p_y} \nonumber
\\
&+ \left( -54\,{\frac {c_{1}\,{\nu}^{2}y}{{
x}^{2}}}+30\,{\frac {{\nu}^{2}{x}^{2}c_{2}}{{y}^{3}}}-108\,{\frac {c_{
2}\,{\nu}^{2}}{y}} \right. \nonumber \\ &  \left. \qquad \qquad \qquad -6\,{\frac {c_{1}\,c_{2}}{{x}^{2}{y}^{3}}}+270\,{\nu
}^{4}{x}^{2}y+486\,{y}^{3}{\nu}^{4}+6\,{\frac {{c_{2}}^{2}}{{y}^{5}}}
 \right) {{\it p_x}}^{3} \nonumber
 \\
 &+\frac {1}{{\it p_x}} 
 \left( 
 -54\,c_{1}\,{x}^{2
}{\nu}^{6}y-2916\,{y}^{3}{\nu}^{6}c_{1}+2\,{\frac {{\nu}^{6}{x}^{6}c_{
2}}{{y}^{3}}}-540\,{\frac {{\nu}^{6}{x}^{4}c_{2}}{y}} \right. \nonumber \\ & \left. \qquad \qquad \qquad +54\,{\frac {{c_{
1}}^{2}{\nu}^{4}y}{{x}^{2}}} +486\,{\frac {{y}^{3}{\nu}^{4}{c_{1}}^{2}
}{{x}^{4}}} +30\,{\frac {{\nu}^{4}{x}^{4}{c_{2}}^{2}}{{y}^{5}}}
-18\,\frac {{c_{1}}^{3}{\nu}^{2}y}{{x}^{6}} \right) \nonumber
\\
&+ \frac {1}{{\it p_x}} 
\left (
-2\,{\frac {{c_{1}}^{3}c_{2}}{{
x}^{6}{y}^{3}}}+6\,{\frac {{c_{1}}^{2}{c_{2}}^{2}}{{x}^{4}{y}^{5}}}-6
\,{\frac {c_{1}\,{x}^{2}{\nu}^{4}c_{2}}{{y}^{3}}}+648\,{\frac {c_{1}\,
{\nu}^{4}c_{2}}{y}}+6\,{\frac {{\nu}^{2}{c_{1}}^{2}c_{2}}{{x}^{2}{y}^{
3}}} \right. \nonumber \\ & \left. \qquad \qquad \qquad -108\,{\frac {{\nu}^{2}{c_{1}}^{2}c_{2}}{{x}^{4}y}}-36\,{\frac {c_
{1}\,{\nu}^{2}{c_{2}}^{2}}{{y}^{5}}}+18\,{\nu}^{8}{x}^{6}y+
2430\,{x}^{4}{y}^{3}{\nu}^{8} \right) \nonumber
\\&
+ \left( -18\,y{\nu}^{2}
-2\,{\frac {c_{2}
}{{y}^{3}}} \right) {{\it p_x}}^{5} + \left( -54\,y{\nu}^{2}+6\,{
\frac {c_{2}}{{y}^{3}}} \right) {{\it p_y}}^{2}{{\it p_x}}^{3} \nonumber
\\
&+ \left( 540\,{\nu}^{4}{x}^{2}y-60\,
{\frac {{\nu}^{2}{x}^{2}c_{2}}{{y}^{3}}}+12\,{\frac {c_{1}\,c_{2}}{{x}
^{2}{y}^{3}}}-108\,{\frac {c_{1}\,{\nu}^{2}y}{{x}^{2}}} \right) {\it 
p_x}\,{{\it p_y}}^{2} \nonumber \\ &+ \left( -120\,{\nu}^{4}{x}^{3}-1620\,{\nu}^{4}x{y}
^{2}+72\,{\frac {c_{1}\,{\nu}^{2}}{x}}+180\,{\frac {{\nu}^{2}xc_{2}}{{
y}^{2}}} \right) {{\it p_x}}^{2}{\it p_y} \nonumber
\\
&+{\frac {{\it p_y}}{{{\it p_x}}^{
2}} \left( -72\,{\frac {c_{1}\,x{\nu}^{4}c_{2}}{{y}^{2}}}+36\,{\frac {
{\nu}^{2}{c_{1}}^{2}c_{2}}{{x}^{3}{y}^{2}}}-324\,{\nu}^{8}{x}^{5}{y}^{
2} \right. } \nonumber \\ & { \left. \qquad \qquad \qquad -324\,{\frac {{c_{1}}^{2}{\nu}^{4}{y}^{2}}{{x}^{3}}}+648\,c_{1}\,x{
\nu}^{6}{y}^{2}+36\,{\frac {{\nu}^{6}{x}^{5}c_{2}}{{y}^{2}}} \right) } \nonumber
\\
&+{\frac {{{\it p_y}}^{2}}{{\it px}} \left( -36\,{\frac {c_{1}\,{\nu}^{2
}c_{2}}{{y}^{3}}}-270\,{x}^{4}y{\nu}^{6}-54\,{\frac {y{\nu}^{2}{c_{1}}
^{2}}{{x}^{4}}} \right. } \nonumber \\ & { \left. \qquad \qquad \qquad +6\,{\frac {{c_{1}}^{2}c_{2}}{{x}^{4}{y}^{3}}}+324\,y{
\nu}^{4}c_{1}+30\,{\frac {{\nu}^{4}{x}^{4}c_{2}}{{y}^{3}}} \right) }+
36\,{{\it p_x}}^{4}x{\nu}^{2}{\it p_y}. \nonumber
\end{align}


\begin{thebibliography}{99}



\bibitem{AKN} V.I. Arnold, V.V. Kozlov, A.I. Neishtadt \textit{Mathematical aspects of classical and celestial mechanics} Springer, Berlin (1997).

\bibitem{BHKN} A. Ballesteros, F. J. Herranz, S. Kuru and J. Negro \textit{The anisotropic oscillator on curved spaces: a new exactly solvable model}, Ann. Phys. \textbf{373},  399-423 (2016).

\bibitem{Ben80} S. Benenti, \textit{Separability structures on Riemannian manifolds}, Lect. Notes Math. \textbf{836}, 512--538, Springer, Berlin (1980).




\bibitem{Ben93} S. Benenti, \textit{Orthogonal separable dynamical Systems}, Math. Publ., Silesian Univ. Opava \textbf{1} , 163--184 (1993).





\bibitem{BogCMP} O.I. Bogoyavlenskij,  \textit{Necessary Conditions for Existence of Non-Degenerate Hamiltonian Structures}, Commun.  Math. Phys \textbf{182}, 253-290 (1996).

\bibitem{BogI} O.I. Bogoyavlenskij,  \textit{General algebraic identities for the Nijenhuis and Haantjes torsions}, Izvestya Mathematics  \textbf{68}, 1129-1141 (2004).

%
%










%
%

\bibitem{BCRR} R. Brouzet, R. Caboz, J. Rabenivo, V. Ravoson, \textit{Two degrees of freedom quasi--bi--Hamiltonian systems}, J. Phys. A \textbf{29}, 2069--2076 (1996).



\bibitem{CCR} R. Campoamor-Stursberg, J. F. Cari\~nena \& M. F. Ra\~nada,  \textit{Higher-order superintegrability of a Holt related potential}, J. Phys. A \textbf{46}, 435202, 6 pp.  (2013).

%


\bibitem{Chanu2001} C. Chanu, \textit{Separation of variables and Killing tensors in the euclidean three-space}, Ph.D. Thesis, Universit\`a di Torino (2001).

\bibitem{D94} B. Dubrovin, Geometry of 2D topological field theories. In: {\it Integrable Systems and Quantum Groups} (Authors: R. Donagi, B. Dubrovin, E. Frenkel, E. Previato),
Eds. M. Francaviglia, S.Greco, Springer Lecture Notes in Math., {\bf 1620} (1996), pp. 120-348.

\bibitem{Evans1990} N. W. Evans, \textit{Superintegrability in classical mechanics}, Phys. Rev. A, \textbf{41}, 5666 (1990). 



\bibitem{FMT} G. Falqui, F. Magri and G. Tondo, \textit{Bi--Hamiltonian systems and separation of variables: an example from the Boussinesq hierarchy}, Theor. Math. Phys. \textbf{122}, 176--192 (2000).




\bibitem{FP} G. Falqui, M. Pedroni, \textit{Separation of variables for bi-hamiltonian systems}, Math. Phys. Anal. Geom. \textbf{6},  139-179 (2003).

\bibitem{Fasso} F. Fass\`o, \textit{Superintegrable hamiltonian systems: geometry and perturbations}, Acta Appl. Math. \textbf{87}, 93 (2005). 




\bibitem{FeMa} E.V. Ferapontov and D.G. Marshall, \textit{Differential-geometric approach to the integrability of hydrodynamics chains: the Haantjes tensor}, Mat. Ann. \textbf{339}, 61--99 (2007).



\bibitem{FMSUW} J. Fri\v{s}, V. Mandrosov, Ya. A. Smorodinsky, M. Uhli\v{r} and P. Winternitz, \textit{On higher
symmetries in quantum mechanics}, Phys. Lett. \textbf{16}, 354 (1965).

\bibitem{FN} A. Frolicher and  A. Nijenhuis, \textit{Theory of Vector-Valued Differential Forms. Part I}, Indag. Mathematicae \textbf{18}, 338--359 (1956).



%
\bibitem{GVY} V.S. Gerdjikov,  G. Vilasi  and A.B. Yanovski, {\it  Integrable Hamiltonian Hierarchies.} Lect. Not. Phys., Vol. {\bf 748} Berlin Heidelberg: Springer, (2008).



\bibitem{Haa} J. Haantjes, \textit{On $X_{n-1}$-forming sets of eigenvectors}, Indag. Mathematicae \textbf{17}, 158--162 (1955).

\bibitem{IMM} A. Ibort, F. Magri, G. Marmo, Bi-hamiltonian structures and St\"ackel separability, J. Geom. Phys. \textbf{33},  210-228 (2000).



%



\bibitem{Rob}	M. Kreuzer, L. Robbiano. \textit{Computational linear and commutative algebra}, Springer (1993).




\bibitem{L} T. Levi--Civita, \textit{Sulla integrazione dell'equazione di Hamilton--Jacobi per separazione di variabili}, Math. Ann. \textbf{59}, 383--397 (1904).


\bibitem{Magri78} F. Magri, \textit{A simple model of the integrable Hamiltonian equation}, J. Math. Phys. \textbf{19}, no. 5, 1156--1162 (1978).


\bibitem{MagriLE} F. Magri, \textit{A Geometrical approach to nonlinear solvable equations}, in Lecture Notes in physics \textbf{120}, (M. Boiti, F. Pempinelli, G. Soliani eds.) Springer--Verlag, Berlin),  233--263 (1980).


\bibitem{Mnoi} F. Magri, \textit{Geometry and  soliton equations}. In: Atti Accad. Sci. Torino Suppl.  physics \textbf{124},  181--209  (1990).

%





\bibitem{MLenard} F. Magri, \textit{Lenard chains for classical integrable systems}, Theoret. and Math. Phys. \textbf{137}, 1716--1722 (2003).


\bibitem{MFrob} F. Magri, \textit{Recursion operators and Frobenius manifolds}, SIGMA   \textbf{8}, paper 076, 7 pages (2012).

\bibitem{MGall13} F. Magri, \textit{Haantjes manifolds}, J Phys Conf Ser \textbf{482}, paper 012028, 10 pages  (2014).


\bibitem{Mp} F. Magri, \textit{The Kowaleski's top revisited}, In: Vol. 1 of Integrable systems and algebraic geometry (R. Donagi and  T. Shaska eds.) London Math. Soc. Lect. Notes, 329--355 (2020).


\bibitem{MM1984} F. Magri, C. Morosi, \textit{Characterization of Integrable Systems through the Theory of Poisson--Nijenhuis Manifolds}  Quaderno  S \textbf{19}, Universit\`a di Milano (1984).

\bibitem{MSVW} A. A. Makarov, J. A. Smorodinsky, Kh. Valiev and P. Winternitz, \textit{A systematic search for nonrelativistic
systems with dynamical symmetries}, Nuovo Cimento A \textbf{52}, 1061-1084 (1967).





\bibitem{MPW} W. Miller Jr, S. Post and P. Winternitz, \textit{Classical and Quantum Superintegrability with Applications}, J. Phys. A, \textbf{46}, 423001 (2013).


\bibitem{MF} A. S. Mischenko and A. T. Fomenko, \textit{Generalized Liouville method of integration of Hamiltonian systems}, Funct. Anal. Appl. \textbf{12}, 113-121 (1978).




\bibitem{MT} C. Morosi, G. Tondo, \textit{Quasi--bi--Hamiltonian systems and separability}, J . Phys. A \textbf{30}, 2799--2806 (1997).

\bibitem{MTPLA} C. Morosi, G. Tondo, \textit{On a class of dynamical systems both quasi--bi--Hamiltonian and bi--Hamiltonian},  Phys. Lett. A \textbf{247}, 59--64 (1998).

\bibitem{MTlt} C. Morosi, G. Tondo, \textit{The quasi-bi-Hamiltonian formulation of the Lagrange top},
J. Phys. A \textbf{35}, 1741--1750 (2002).

\bibitem{MTltC} C. Morosi, G. Tondo, \textit{Separation of Variables in multi--Hamiltonian systems: an application to  the Lagrange top},
Theor. Math. Phys. \textbf{137}, 1550--1560, (2003).

\bibitem{Nekh} N. N. Nekhoroshev, \textit{Action-angle variables and their generalizations}, Trans. Moscow Math. Soc. \textbf{26}, 180-198 (1972).



\bibitem{NN} A. Newlander, L. Nirenberg, \textit{Complex analytic coordinates in almost complex manifolds}, Ann. Math. \textbf{65},  391-404 (1957).


\bibitem{Nij} A. Nijenhuis, \textit{$X_{n-1}$-forming sets of eigenvectors}, Indag. Mathematicae  \textbf{54}, 200-212 (1951).

\bibitem{Nij2} A. Nijenhuis, \textit{Jacobi--type identities for bilinear differential concomitants of certain tensor fields I,II},
Indag. Math \textbf{17}, 390--397, 398--403 (1955).









\bibitem{PW2011} S. Post, P. Winternitz, \textit{A non separable quantum superintegrable system in 2D real Euclidean space} J. Phys. A \textbf{44}, 162001, 8pp. (2011)






\bibitem{RTW1} M. A. Rodr\'{\i}guez, p. Tempesta and p. Winternitz, \textit{Reduction of superintegrable systems: the anisotropic harmonic oscillator}, Phys. Rev. E \textbf{78}, 046608 (2008).

\bibitem{RTW2} M. A. Rodr\'{\i}guez, P. Tempesta and P. Winternitz, \textit{Symmetry reduction and superintegrable Hamiltonian systems}, J. Phys. Conf. Series \textbf{175}, 012013 (2009).

\bibitem{Rosochatius} E. Rosochatius, Dissertation, Gottingen, Gebr. Unger, Berlin 1877.



\bibitem{STW} M. Sheftel, P. Tempesta and P. Winternitz, \textit{Superintegrable Systems in Quantum Mechanics and Classical Lie Theory}, J. Math. Phys. \tbf{42}, 659 (2001). 

\bibitem{Skl} E. K. Sklyanin, \textit{Separation of variables: new trends}, Prog. Theor. Phys. Suppl. \textbf{118}, 35--60 (1995).


\bibitem{TT2012} P. Tempesta, G. Tondo, \textit{Generalized Lenard chains, separation and Superintegrability}, Phys. Rev E \textbf{85}, paper n. 046602, 11  pages (2012)



\bibitem{TT2019prepr} P. Tempesta, G. Tondo, \textit{Higher Haantjes brackets and integrability}, accepted in Comm. Math. Phys., Preprint arXiv: 1809.05908v5, (2021).

\bibitem{TT2021} P. Tempesta, G. Tondo, \textit{Haantjes Algebras and Diagonalization}. J. Geom. Phys. \textbf{160}, Paper n. 103968 (2021).



\bibitem{TT2016prepr} P. Tempesta, G. Tondo, \textit{Haantjes Algebras of Classical Integrable Systems},  Ann. Math. Pura Appl., https://doi.org/10.1007/s10231-021-01107-4 (2021), Preprint arXiv: 1405.5118v3.


\bibitem{TTW} P. Tempesta, A. Turbiner and P. Winternitz, \textit{Exact
Solvability of Superintegrable Systems}, J. Math. Phys. \textbf{42}, 4248 (2001).

\bibitem{TWR} P. Tempesta, P. Winternitz, J. Harnad, W. Miller, Jr, G. Pogosyan, M. A. Rodr\'{\i}guez (eds), \textit{Superintegrability in Classical and Quantum Systems}, Montr\'{e}al, CRM proceedings and Lecture Notes, AMS, vol. \textbf{37} (2004).





\bibitem{TGalli12} G. Tondo, \textit{Generalized Lenard chains and multi--separability of the Smorodinsky-Winternitz system},
J. Phys.: Conference Series \textbf{482}, 012042, 10 pp. (2014).

\bibitem{T2017}  G. Tondo, \textit{Haantjes algebras of the Lagrange top}, Theor. Math. Phys. \textbf{196}, 1366--1379 (2018).


\bibitem{MTRomp} G. Tondo, C. Morosi,  \textit{Bi-Hamiltonian manifolds,  quasi--bi--Hamiltonian systems and Separation of Variables}, Rep. Math. Phys. \textbf{44}, 255--266 (1999).

\bibitem{TT2016SIGMA}  G. Tondo, P. Tempesta, \textit{Haantjes structures for the Jacobi-Calogero model and the Benenti Systems}, SIGMA \textbf{12}, paper 023, 18 pp. (2016).



\bibitem{YI1973} K. Yano and S. Ishihara, Tangent and Cotangent bundles, Dekker, Pure and Applied Mathematics vol. 16, 1973. 
  

\bibitem{Winternitz} P. Winternitz, Ya. Smorodinsky, M. Uhli\v{r} and J. Fri\v{s}, Yad. Fiz. \textbf{4}, 625 (1966), \textit{Symmetry groups in classical and quantum mechanics}, Sov. J. Nucl. Phys \textbf{4}, 444 (1967).


\bibitem{Wojc1985} S. Wojciechowski, \textit{Integrability of one particle in a perturbed central quartic potential}, Phys. Scr. \tbf{31}, 433-438 (1985).




\end{thebibliography}
\end{document}